\documentclass[draftcls,onecolumn,11pt]{IEEEtran}    
\usepackage[T1]{fontenc}

\usepackage{amsmath}
\usepackage{amssymb}
\usepackage{stmaryrd}
\usepackage{epic,eepic}
\usepackage{theorem}
\usepackage{pifont}  
\usepackage{euscript}
\usepackage{calc}
\usepackage[titles]{tocloft}

\usepackage[usenames]{color}          
\usepackage{xcolor}
\definecolor{Brown}{rgb}{0.55,0.0,0.10}
\definecolor{dgreen}{rgb}{0.00,0.56,0.00}
\definecolor{vertmoinsfonce}{rgb}{0.00,0.50,0.00}
\definecolor{vert}{rgb}{0.00,0.60,0.00}
\definecolor{llightggray}{rgb}{0.97,0.97,0.97}
\definecolor{lightggray}{rgb}{0.9,0.9,0.9}
\definecolor{ggray}{rgb}{0.5,0.5,0.5}
\definecolor{darkggray}{rgb}{0.25,0.25,0.25}
\definecolor{ddarkggray}{rgb}{0.1,0.1,0.1}
\definecolor{bleu}{rgb}{0.00,0.00,1.00}
\definecolor{darkblue}{rgb}{0,0,0.7}

\usepackage{epsf}           
\usepackage{graphicx}       
\usepackage{epsfig}         
\usepackage{pstricks}
\usepackage{pstricks-add}
\usepackage{pst-plot}
\usepackage{dsfont}
\usepackage{pst-circ}
\usepackage{pst-math}

\theoremheaderfont{\color{black}\normalfont\bfseries}

\newtheorem{lemma}{Lemma}
\newtheorem{theorem}{Theorem}[section]
\newtheorem{definition}[theorem]{Definition}
\newtheorem{corollary}[theorem]{Corollary}
\newtheorem{proposition}[theorem]{Proposition}

\theoremstyle{plain}{\theorembodyfont{\rmfamily}%
}
\theoremstyle{plain}{\theorembodyfont{\rmfamily}%
}
\theoremstyle{plain}{
\theorembodyfont{\rmfamily}

	\newtheorem{remark}[theorem]{Remark}

	}

\newcommand{\cav}{\ensuremath{\operatorname{cav}}}

\newcommand{\argmin}{\ensuremath{\operatorname{argmin}}}

\newcommand{\argmax}{\ensuremath{\operatorname{argmax}}}

\newcommand{\R}{\mathbb{R}}

\newcommand{\N}{\mathbb{N}}

\newcommand{\E}{\mathbb{E}}
\newcommand{\Q}{\mathbb{Q}}

\font\dsrom=dsrom10 scaled 1200 \def \indic{\textrm{\dsrom{1}}}
\newcommand{\UN}{\indic}

\newcommand{\DD}{\mathcal{D}}
\newcommand{\QQ}{\mathcal{Q}}

\newcommand{\PP}{\mathcal{P}}
\newcommand{\EE}{\mathcal{E}}
\newcommand{\mc}{\mathcal}

{\Large\normalsize}
{\Large\normalsize}

\usepackage[nomessages]{fp}
\newcommand\const[3][6]{%
    \edef\temporary{round(#3}%
    \expandafter\FPeval\csname#2\expandafter\endcsname
        \expandafter{\temporary:#1)}%
        \pstVerb{/#2 \csname#2\endcsname\space def}%
}

\newcommand{\FigDynamic}[4]{

\const{qqA}{#1 * #3 / (#1 * #3 + (1 - #1) * (1 - #2))}
\const{qqB}{#1 * (1-#3) / (#1 * (1-#3) + (1 - #1) * #2)}
\const{nuA}{#4 * #1  * (1 - #2)/ (\qqA * (1 - #2) + (\qqA - #4) * #1 * (#2 + #3 - 1))}
\const{nuB}{#4 * #1  * #2/ (\qqB * #2 + (\qqB - #4) * #1 * (1 - #2 - #3))}

\const{pBnuA}{#4 * (1-#2)  * \qqB * ((1  -#1) * #2 + #1 * (1 - #3)) /((\qqA * (1 - #1) * (1 - #2) + (\qqA - #4) * #1 * #3 )* #2 + #4 * #1 * ( 1 - #2) * ( 1 - #3))}
\const{pAnuB}{#4 * #2  * \qqA * ((1  -#1) * (1 - #2) + #1 *  #3) /((\qqB * (1 - #1) * #2 + (\qqB - #4) * #1 * (1 - #3) )* (1 - #2) + #4 * #1 *  #2 *  #3)}

\def\ppA{x \qqA \space mul  #1 neg 1 add #2 neg 1 add mul #1 #3 mul add mul #1  x neg 1 add #2 neg 1 add mul x #3 mul add mul div} 
\def\ppB{x \qqB \space mul  #1 neg 1 add #2 mul #1 #3 neg 1 add mul add mul #1  x neg 1 add #2 mul x #3 neg 1 add mul add mul div}

\begin{figure}[!h]
\begin{center}
\psset{xunit=5cm,yunit=5cm}
\begin{pspicture}(-0.2,-0.28)(1.2,1.2)
\psline{->}(0,-0.1)(0,1.2)
\psline{->}(-0.1,0)(1.2,0)
\psline{-}(1,-0.1)(1,1.2)
\psline{-}(-0.1,1)(1.1,1)

\rput[r](-0.02,-0.05){$0$}
\rput[r](-0.02,1.04){$1$}
\rput[u](1.03,-0.05){$1$}

\rput[l]{-45}(#1,-0.06){$p_0 = #1$}
\psline{-}(#1,0)(#1,1.2)
\psline[linecolor=blue,linestyle=dotted]{-}(0,#1)(1,#1)
\rput[r](-0.02,#1){$p_0 = #1$}
\psline[linecolor=blue,linestyle=dotted]{-}(0,\qqA)(1,\qqA)
\rput[r](-0.02, \qqA){$p_{01} = \qqA$}
\psline[linecolor=blue,linestyle=dotted]{-}(0,\qqB)(1,\qqB)
\rput[r](-0.02, \qqB){$p_{02} = \qqB$}
\Aput*{$p_1$}
\psplot[plotpoints=100,linecolor=blue]{0}{#1}{ \ppA}
\psplot[plotpoints=100,linecolor=green]{#1}{1}{ \ppA}
\psplot[plotpoints=100,linecolor=brown]{0}{#1}{ \ppB}
\psplot[plotpoints=100,linecolor=red]{#1}{1}{ \ppB}
\psline{-}(0,0)(1,1)
\psline[linecolor=blue,linestyle=dashed]{-}(\nuA,0)(\nuA,1)
\psline[linecolor=blue,linestyle=dashed]{-}(\nuB,0)(\nuB,1)
\psline[linecolor=blue,linestyle=dashed]{-}(0,#4)(1,#4)
\rput[l](1.04, #4){$\gamma = #4$}
\rput[l]{-45}(\nuA,-0.06){$\nu_1 = \nuA$}
\rput[l]{-45}(\nuB,-0.06){$\nu_2 = \nuB$}
\psdots[linecolor=blue](\nuA,0)(\nuB,0)(\nuA,#4)(\nuB,#4)
\pcline{<->}(0,1.12)(#1,1.12)\Aput{$q_1$}
\pcline{<->}(#1,1.12)(1,1.12)\Aput{$q_2$}
\psline[linecolor=orange,linestyle=dashed]{-}(1,\pBnuA)(\nuA,\pBnuA)
\psline[linecolor=orange,linestyle=dashed]{-}(1,\pAnuB)(\nuB,\pAnuB)
\psdots[linecolor=orange](\nuA,\pBnuA)(\nuB,\pAnuB)(1,\pBnuA)(1,\pAnuB)
\rput[l](1.04, \pAnuB){$p_3(\nu_2) = \pAnuB$}
\rput[l](1.04, \pBnuA){$p_2(\nu_1) = \pBnuA$}

\rput[u](! #1 \space 2 \space div  #1 \space 0.055 \space add ){$p_1$}
\rput[u](! #1 \space 2 \space div  \qqB \space 0.035 \space neg \space add ){$p_2$}
\rput[u](! #1 \space  2 \space div \space neg \space 1 \space add \qqA \space 0.1 \space add ){$p_3$}
\rput[u](! #1 \space  2 \space div \space neg \space 1 \space add \qqB \space 0.1 \space add ){$p_4$}

\end{pspicture}
\caption{Posterior beliefs $(p_1,p_2)$ depending on $q_1\in[0,p_0]$ and posterior beliefs $(p_3,p_4)$ depending on $q_2\in[p_0,1]$, for $p_0=#1$, $\delta_1=#2$, $\delta_2=#3$ and $\gamma = #4$.}\label{fig:Dynamic#1_#2_#3_#4}
\end{center}
\end{figure}}

\newcommand{\PSPlayerOneAuxUtility}[4]{
\const{qqA}{#1 * #3 / (#1 * #3 + (1 - #1) * (1 - #2))}
\const{qqB}{#1 * (1-#3) / (#1 * (1-#3) + (1 - #1) * #2)}
\const{nuA}{#4 * #1  * (1 - #2)/ (\qqA * (1 - #2) + (\qqA - #4) * #1 * (#2 + #3 - 1))}
\const{nuB}{#4 * #1  * #2/ (\qqB * #2 + (\qqB - #4) * #1 * (1 - #2 - #3))}
\const{PsiANuA}{(1- \nuA ) * (1 - #2) + \nuA * #3}
\const{PsiANuB}{(1- \nuB ) * (1 - #2) + \nuB * #3}
\const{PsiBNuA}{(1- \nuA ) * #2 + \nuA * (1 - #3)}
\const{PsiBNuB}{(1- \nuB ) * #2 + \nuB * (1 - #3)}
\const{PsiStarA}{(1 - \PsiANuA) * (#1 - \nuA) / (\nuB - \nuA) + \PsiANuA}
\def\ppA{x \qqA \space mul  #1 neg 1 add #2 neg 1 add mul #1 #3 mul add mul #1  x neg 1 add #2 neg 1 add mul x #3 mul add mul div} 
\def\ppB{x \qqB \space mul  #1 neg 1 add #2 mul #1 #3 neg 1 add mul add mul #1  x neg 1 add #2 mul x #3 neg 1 add mul add mul div}
\psline{->}(0,-0.1)(0,1.2)
\psline{->}(-0.1,0)(1.2,0)
\psline{-}(1,-0.1)(1,1.2)
\rput[u](-0.03,-0.03){$0$}
\rput[l](1.05,1){$1$}
\rput[u](1.05,0.05){$1$}
\rput[u](1.15,-0.05){$q$}
\psline[](#1,0)(#1,1.2)
\rput[l]{-45}(#1,-0.04){$ p_0 = #1$}
\psline[linecolor=orange, linewidth=3.5pt](0,0)(\nuA,0)
\psline[linecolor=orange, linewidth=3.5pt](\nuB,1)(1,1)
\rput[l]{-45}(\nuA,-0.045){$\nu_1 = \nuA$}
\rput[l]{-45}(\nuB,-0.045){$\nu_2 = \nuB$}
\psdots[linecolor=blue](\nuA,0)(\nuB,0)
\psdots[linecolor=blue](0,0)(\nuB,1)
\psline[linestyle=dotted,linecolor=blue](\nuA,0)(\nuA,1.2)
\psline[linestyle=dotted,linecolor=blue](\nuB,0)(\nuB,1.2)
\psline[linecolor=orange, linewidth=3.5pt](\nuA,\PsiANuA)(\nuB,\PsiANuB)
\psdots[linecolor=blue](\nuA,\PsiANuA)(\nuB,\PsiANuB)
\psline[linestyle=dashed,linecolor=blue](0,0)(\nuA,\PsiANuA)(\nuB,1)
\psline[linestyle=dashed,linecolor=red](#1,\PsiStarA)(1,\PsiStarA)
\psdots[linecolor=red](#1,\PsiStarA)(1,\PsiStarA)
\rput[l](1.02, \PsiStarA){$\Phi_{\textsf{e}}^{\circ} = \PsiStarA$}
\psdots[linecolor=red,dotscale=2.5,dotstyle = square*](\nuA,\PsiANuA)
\psdots[linecolor=red,dotscale=2.5](\nuB,1)
\uput[u]{45}(\nuB,1.05){$(\textcolor[rgb]{0,0.6,0}{v_1^{\star}},\textcolor[rgb]{0,0.6,0}{v_1^{\star}})$}
\uput[u]{45}(0.15,1.05){$(\textcolor[rgb]{0,0,1}{v_2^{\star}},\textcolor[rgb]{0,0,1}{v_2^{\star}})$}
\uput[u]{45}(0.55,1.05){$(\textcolor[rgb]{0,0,1}{v_2^{\star}},\textcolor[rgb]{0,0.6,0}{v_1^{\star}})$}
}

\newcommand{\FigPlayerOneAuxUtility}[4]{
\begin{figure}[!h]
\begin{center}
\psset{xunit=5cm,yunit=5cm}
\begin{pspicture}(0,-0.35)(1,1.2)
\PSPlayerOneAuxUtility{#1}{#2}{#3}{#4}
\end{pspicture}
\caption{Optimal encoder's utility depending on the belief parameter $q\in [0,1]$ after observing $W$, for parameters $p_0=#1$, $\delta_1=#2$, $\delta_2=#3$ and $\gamma = #4$.}\label{fig:PlayerOneAuxUtility#1_#2_#3_#4}
\end{center}
\end{figure}}

\newcommand{\FigPlayerOneAuxUtilityEntropyBis}[5]{
\const{qqA}{#1 * #3 / (#1 * #3 + (1 - #1) * (1 - #2))}
\const{qqB}{#1 * (1-#3) / (#1 * (1-#3) + (1 - #1) * #2)}
\const{nuA}{#4 * #1  * (1 - #2)/ (\qqA * (1 - #2) + (\qqA - #4) * #1 * (#2 + #3 - 1))}
\const{nuB}{#4 * #1  * #2/ (\qqB * #2 + (\qqB - #4) * #1 * (1 - #2 - #3))}
\const{PsiANuA}{(1- \nuA ) * (1 - #2) + \nuA * #3}
\const{PsiANuB}{(1- \nuB ) * (1 - #2) + \nuB * #3}
\const{PsiBNuA}{(1- \nuA ) * #2 + \nuA * (1 - #3)}
\const{PsiBNuB}{(1- \nuB ) * #2 + \nuB * (1 - #3)}
\const{PsiStarA}{(1 - \PsiANuA) * (#1 - \nuA) / (\nuB - \nuA) + \PsiANuA}
\def\ppA{x \qqA \space mul  #1 neg 1 add #2 neg 1 add mul #1 #3 mul add mul #1  x neg 1 add #2 neg 1 add mul x #3 mul add mul div} 
\def\ppB{x \qqB \space mul  #1 neg 1 add #2 mul #1 #3 neg 1 add mul add mul #1  x neg 1 add #2 mul x #3 neg 1 add mul add mul div}

\begin{figure}[!h]
\begin{center}
\psset{xunit=5cm,yunit=5cm}
\begin{pspicture}(0,-0.35)(1,1.3)
\psline{->}(0,-0.1)(0,1.2)
\psline{->}(-0.1,0)(1.2,0)
\psline{-}(1,-0.1)(1,1.2)
\rput[u](-0.03,-0.03){$0$}
\rput[l](1.05,1){$1$}
\rput[u](1.05,0.05){$1$}
\rput[u](1.15,-0.05){$q$}
\rput[l]{-45}(#1,-0.04){$ p_0 = #1$}
\psline[linecolor=orange, linewidth=3.5pt](0,0)(\nuA,0)
\psline[linecolor=orange, linewidth=3.5pt](\nuB,1)(1,1)
\rput[l]{-45}(\nuA,-0.045){$\nu_1 = \nuA$}
\rput[l]{-45}(\nuB,-0.045){$\nu_2 = \nuB$}
\psdots[linecolor=blue](\nuA,0)(\nuB,0)
\psdots[linecolor=blue](\nuB,1)
\psline[](#1,0)(#1,1.2)
\psline[linestyle=dotted,linecolor=blue](\nuA,0)(\nuA,\PsiANuA)
\psline[linestyle=dotted,linecolor=blue](\nuB,0)(\nuB,1.2)
\psline[linecolor=orange, linewidth=3.5pt](\nuA,\PsiANuA)(\nuB,\PsiANuB)
\psdots[linecolor=blue](\nuA,\PsiANuA)(\nuB,\PsiANuB)
\psdots[linecolor=red,dotscale=2.5](\nuB,1)
\psline[linestyle=dotted,linecolor=blue](0.383,0)(0.383,1.2)
\psdots[linecolor=red,dotscale=2.5,dotstyle = square*](0.383 ,0.53)
\psline[linestyle=dashed,linecolor=red](0.383 ,0.53)(\nuB ,1)
\psline[linestyle=dashed,linecolor=red](#1,0.6337)(1,0.6337)
\psdots[linecolor=red](#1,0.6337)(1,0.6337)
\rput[l](1.02, 0.6337){$\Phi_{\textsf{e}}^{\star} = 0.6337$}
\rput[l]{-45}(0.41,-0.045){$\nu_3 = 0.383$}
\uput[u]{45}(\nuB,1.05){$(\textcolor[rgb]{0,0.6,0}{v_1^{\star}},\textcolor[rgb]{0,0.6,0}{v_1^{\star}})$}
\uput[u]{45}(0.15,1.05){$(\textcolor[rgb]{0,0,1}{v_2^{\star}},\textcolor[rgb]{0,0,1}{v_2^{\star}})$}
\uput[u]{45}(0.55,1.05){$(\textcolor[rgb]{0,0,1}{v_2^{\star}},\textcolor[rgb]{0,0.6,0}{v_1^{\star}})$}
\def\entropyU{#1 \space ln 2 ln div #1 \space mul neg 1 #1 \space neg add ln 2 ln div 1 #1 \space neg add mul neg add}
\def\entropyDA{#2 \space  ln 2 ln div #2 \space mul neg 1 #2 \space neg add ln 2 ln div 1 #2 \space neg add mul neg add}
\def\entropyDB{#3 \space  ln 2 ln div #3 \space mul neg 1 #3 \space neg add ln 2 ln div 1 #3 \space neg add mul neg add}
\def\entropyZsU{#1 neg 1 add \entropyDA \space mul #1 \entropyDB \space mul add}
\def\probaZ{#1 neg 1 add #2 mul #3 neg 1 add #1 mul add}
\def\entropyZ{\probaZ \space  ln 2 ln div \probaZ \space mul neg 1 \probaZ \space neg add ln 2 ln div 1 \probaZ \space neg add mul neg add}
\def\IC{\entropyU \space \entropyZsU \space add \entropyZ \space neg add #5 neg add}
\psline[linecolor=black](! 0 \IC)(! 1 \IC)
\psdots[linecolor=black](! #1 \IC)(! 0 \IC)
\rput[r](!-0.02 \IC){$H(U|Z) - C$}
\def\lambdaA{x neg 1 add #2 neg 1 add mul x #3 mul add}
\def\lambdaB{x neg 1 add #2 mul #3 neg 1 add x mul add}
\def\entropyA{\ppA \space ln 2 ln div \ppA \space mul neg 1 \ppA \space neg add ln 2 ln div 1 \ppA \space neg add mul neg add}
\def\entropyB{\ppB \space  ln 2 ln div \ppB \space mul neg 1 \ppB \space neg add ln 2 ln div 1 \ppB \space neg add mul neg add}
\def\HUZw{\lambdaA \space \entropyA \space mul \lambdaB\space  \entropyB \space mul add}
\psplot[plotpoints=100]{0.001}{0.999}{\HUZw}
\def\lambdaANuA{\nuA \space neg 1 add #2 neg 1 add mul \nuA \space  #3 mul add}
\def\lambdaBNuA{\nuA \space  neg 1 add #2 mul #3 neg 1 add \nuA \space  mul add}
\def\ppANuA{\nuA \space \qqA \space mul  #1 neg 1 add #2 neg 1 add mul #1 #3 mul add mul #1  \nuA \space neg 1 add #2 neg 1 add mul \nuA \space #3 mul add mul div} 
\def\ppBNuA{\nuA \space \qqB \space mul  #1 neg 1 add #2 mul #1 #3 neg 1 add mul add mul #1  \nuA \space neg 1 add #2 mul \nuA \space #3 neg 1 add mul add mul div}
\def\entropyANuA{\ppANuA \space ln 2 ln div \ppANuA \space mul neg 1 \ppANuA \space neg add ln 2 ln div 1 \ppANuA \space neg add mul neg add}
\def\entropyBNuA{\ppBNuA \space  ln 2 ln div \ppBNuA \space mul neg 1 \ppBNuA \space neg add ln 2 ln div 1 \ppBNuA \space neg add mul neg add}
\def\HUZwNuA{\lambdaANuA \space \entropyANuA \space mul \lambdaBNuA \space  \entropyBNuA \space mul add}
\def\lambdaANuB{\nuB \space neg 1 add #2 neg 1 add mul \nuB \space  #3 mul add}
\def\lambdaBNuB{\nuB \space  neg 1 add #2 mul #3 neg 1 add \nuB \space  mul add}
\def\ppANuB{\nuB \space \qqA \space mul  #1 neg 1 add #2 neg 1 add mul #1 #3 mul add mul #1  \nuB \space neg 1 add #2 neg 1 add mul \nuB \space #3 mul add mul div} 
\def\ppBNuB{\nuB \space \qqB \space mul  #1 neg 1 add #2 mul #1 #3 neg 1 add mul add mul #1  \nuB \space neg 1 add #2 mul \nuB \space #3 neg 1 add mul add mul div}
\def\entropyANuB{\ppANuB \space ln 2 ln div \ppANuB \space mul neg 1 \ppANuB \space neg add ln 2 ln div 1 \ppANuB \space neg add mul neg add}
\def\entropyBNuB{\ppBNuB \space  ln 2 ln div \ppBNuB \space mul neg 1 \ppBNuB \space neg add ln 2 ln div 1 \ppBNuB \space neg add mul neg add}
\def\HUZwNuB{\lambdaANuB \space \entropyANuB \space mul \lambdaBNuB \space  \entropyBNuB \space mul add}
\psline[linecolor=blue](! \nuB \space \HUZwNuB)(0.383 ,0.685)
\psdots[linecolor=blue](0.383,0)(0.383,0.685)(! \nuB \space \HUZwNuB)
\psdots[linecolor=blue](! #1 \IC)(! 0 \IC)

\def\HUZwNuANuB{\HUZwNuA \space  \HUZwNuB \space neg add \nuB \space #1 neg add mul \nuB \space \nuA \space neg add div \HUZwNuB \space add}

\end{pspicture}
\caption{Optimal encoder's utility depending on the belief parameter $q\in [0,1]$ after observing $W$, for parameters $p_0=#1$, $\delta_1=#2$, $\delta_2=#3$, $\gamma = #4$ and $C = #5$. The curve represents the average entropy $h(q)= H(U|Z)$ defined in \eqref{eq:AverageEntropy}, as a function of $q\in[0,1]$.}\label{fig:PlayerOneAuxUtilityEntropyBis#1_#2_#3_#4_#5}
\end{center}
\end{figure}}

\newcommand{\PSPlayerOneUtility}[4]{
\const{qqA}{#1 * #3 / (#1 * #3 + (1 - #1) * (1 - #2))}
\const{qqB}{#1 * (1-#3) / (#1 * (1-#3) + (1 - #1) * #2)}
\const{nuA}{#4 * #1  * (1 - #2)/ (\qqA * (1 - #2) + (\qqA - #4) * #1 * (#2 + #3 - 1))}
\const{nuB}{#4 * #1  * #2/ (\qqB * #2 + (\qqB - #4) * #1 * (1 - #2 - #3))}
\const{PsiANuA}{(1- \nuA ) * (1 - #2) + \nuA * #3}
\const{PsiANuB}{(1- \nuB ) * (1 - #2) + \nuB * #3}
\const{PsiBNuA}{(1- \nuA ) * #2 + \nuA * (1 - #3)}
\const{PsiBNuB}{(1- \nuB ) * #2 + \nuB * (1 - #3)}
\const{PsiStarO}{ #1 / #4}
\const{PsiStarA}{(1 - \PsiANuA) * (#1 - \nuA) / (\nuB - \nuA) + \PsiANuA}
\const{pBnuA}{#4 * (1-#2)  * \qqB * ((1  -#1) * #2 + #1 * (1 - #3)) /((\qqA * (1 - #1) * (1 - #2) + (\qqA - #4) * #1 * #3 )* #2 + #4 * #1 * ( 1 - #2) * ( 1 - #3))}
\const{pAnuB}{#4 * #2  * \qqA * ((1  -#1) * (1 - #2) + #1 *  #3) /((\qqB * (1 - #1) * #2 + (\qqB - #4) * #1 * (1 - #3) )* (1 - #2) + #4 * #1 *  #2 *  #3)}
\const{PsiStarqqA}{(\qqB - \pBnuA) / (#4 - \pBnuA)}
\def\ppA{x \qqA \space mul  #1 neg 1 add #2 neg 1 add mul #1 #3 mul add mul #1  x neg 1 add #2 neg 1 add mul x #3 mul add mul div} 
\def\ppB{x \qqB \space mul  #1 neg 1 add #2 mul #1 #3 neg 1 add mul add mul #1  x neg 1 add #2 mul x #3 neg 1 add mul add mul div}

\psline{->}(0,-0.1)(0,1.2)
\psline{->}(-0.1,0)(1.2,0)
\psline{-}(1,-0.1)(1,1.2)
\rput[u](-0.03,-0.03){$0$}
\rput[l](1.05,1){$1$}
\rput[u](1.05,0.05){$1$}
\rput[u](1.15,-0.05){$p(u_2)$}
\psline(#1,0)(#1,1.2)
\rput[l]{45}(#1,1.24){$ p_0 = #1$}
\psdots(#1,0)
\psline[linecolor=orange, linewidth=3.5pt](0,0)(#4,0)
\psline[linecolor=orange, linewidth=3.5pt](#4,1)(1,1)
\psline[linestyle=dotted,linecolor=orange](#4,0)(#4,1)
\psline[linestyle=dashed,linecolor=orange](0,0)(#4,1)
\psdots[linecolor=orange](#1,\PsiStarO)
\psline[linestyle=dashed,linecolor=orange](#1,\PsiStarO)(1,\PsiStarO)
\rput[l](1.02, \PsiStarO){$\widetilde{\Phi_{\textsf{e}}} = \PsiStarO$}
\uput[u](0.8,1.05){$\textcolor[rgb]{0,0.6,0}{v_1^{\star}}$}
\uput[u](0.36,1.05){$\textcolor[rgb]{0,0,1}{v_2^{\star}}$}
\rput[l]{45}(\qqA,1.24){$\textcolor[rgb]{0,0,1}{p_{01}= \qqA}$}
\psline[linestyle=dotted,linecolor=blue](\qqA,0)(\qqA,1.2)
\psdots[linecolor=blue](\qqA,0)(\qqA,1)(\pAnuB,0)
\psdots[linecolor=blue,dotscale=2](\pAnuB,1)
\psline[linestyle=dotted,linecolor=blue](\pAnuB,0)(\pAnuB,1.2)
\psdots[linecolor=blue,dotscale=2.5, dotstyle=square*](#4,1)
\rput[l]{45}(\qqB,1.24){$\textcolor[rgb]{1,0,0}{p_{02}=\qqB}$}
\psline[linestyle=dotted,linecolor=red](\qqB,0)(\qqB,1.2)
\rput[l]{-45}(\pBnuA,-0.045){$p_2=p_2(\nu_1)=\pBnuA$}
\rput[l]{-45}(\pAnuB,-0.045){$p_3=p_1(\nu_2)=\pAnuB$}
\rput[l]{-45}(#4,-0.045){$p_1 = p_4 = \gamma=#4$}
\psdots[linecolor=red,dotscale=2](#4,1)
\psdots[linecolor=red,dotscale=2.5, dotstyle=square*](\pBnuA,0)
\psdots[linecolor=red](\qqB,0)(\qqB,\PsiStarqqA)
\psline[linestyle=dashed,linecolor=red](\pBnuA,0)(#4,1)

\psline[linestyle=dashed,linecolor=brown](\qqB,\PsiStarqqA)(\qqA,1)
\psline[linestyle=dashed,linecolor=brown](#1,\PsiStarA)(1,\PsiStarA)
\psdots[linecolor=brown](#1,\PsiStarA)(1,\PsiStarA)
\rput[l](1.02, \PsiStarA){$\Phi_{\textsf{e}}^{\circ} = \PsiStarA$}
}

\newcommand{\FigPlayerOneUtility}[4]{
\begin{figure}[!h]
\begin{center}
\psset{xunit=5cm,yunit=5cm}
\begin{pspicture}(0,-0.5)(1,1.5)
\PSPlayerOneUtility{#1}{#2}{#3}{#4}
\end{pspicture}
\caption{Optimal encoder's utility depending on the belief $p(u_2)\in [0,1]$ after observing $(W,Z)$, for parameters $p_0=#1$, $\delta_1=#2$, $\delta_2=#3$ and $\gamma = #4$.}\label{fig:PlayerOneUtility#1_#2_#3_#4}
\end{center}
\end{figure}}

\newcommand{\FigPlayerOneUtilityEntropy}[5]{
\begin{figure}[!h]
\begin{center}
\psset{xunit=5cm,yunit=5cm}
\begin{pspicture}(0,-0.5)(1,1.5)

\const{qqA}{#1 * #3 / (#1 * #3 + (1 - #1) * (1 - #2))}
\const{qqB}{#1 * (1-#3) / (#1 * (1-#3) + (1 - #1) * #2)}
\const{nuA}{#4 * #1  * (1 - #2)/ (\qqA * (1 - #2) + (\qqA - #4) * #1 * (#2 + #3 - 1))}
\const{nuB}{#4 * #1  * #2/ (\qqB * #2 + (\qqB - #4) * #1 * (1 - #2 - #3))}
\const{PsiANuA}{(1- \nuA ) * (1 - #2) + \nuA * #3}
\const{PsiANuB}{(1- \nuB ) * (1 - #2) + \nuB * #3}
\const{PsiBNuA}{(1- \nuA ) * #2 + \nuA * (1 - #3)}
\const{PsiBNuB}{(1- \nuB ) * #2 + \nuB * (1 - #3)}
\const{PsiStarO}{ #1 / #4}
\const{PsiStarA}{(1 - \PsiANuA) * (#1 - \nuA) / (\nuB - \nuA) + \PsiANuA}
\const{pBnuA}{#4 * (1-#2)  * \qqB * ((1  -#1) * #2 + #1 * (1 - #3)) /((\qqA * (1 - #1) * (1 - #2) + (\qqA - #4) * #1 * #3 )* #2 + #4 * #1 * ( 1 - #2) * ( 1 - #3))}
\const{pAnuB}{#4 * #2  * \qqA * ((1  -#1) * (1 - #2) + #1 *  #3) /((\qqB * (1 - #1) * #2 + (\qqB - #4) * #1 * (1 - #3) )* (1 - #2) + #4 * #1 *  #2 *  #3)}
\const{PsiStarqqA}{(\qqB - \pBnuA) / (#4 - \pBnuA)}
\def\ppA{x \qqA \space mul  #1 neg 1 add #2 neg 1 add mul #1 #3 mul add mul #1  x neg 1 add #2 neg 1 add mul x #3 mul add mul div} 
\def\ppB{x \qqB \space mul  #1 neg 1 add #2 mul #1 #3 neg 1 add mul add mul #1  x neg 1 add #2 mul x #3 neg 1 add mul add mul div}
\def\entropyU{#1 \space ln 2 ln div #1 \space mul neg 1 #1 \space neg add ln 2 ln div 1 #1 \space neg add mul neg add}
\def\entropyDA{#2 \space  ln 2 ln div #2 \space mul neg 1 #2 \space neg add ln 2 ln div 1 #2 \space neg add mul neg add}
\def\entropyDB{#3 \space  ln 2 ln div #3 \space mul neg 1 #3 \space neg add ln 2 ln div 1 #3 \space neg add mul neg add}
\def\entropyZsU{#1 neg 1 add \entropyDA \space mul #1 \entropyDB \space mul add}
\def\probaZ{#1 neg 1 add #2 mul #3 neg 1 add #1 mul add}
\def\entropyZ{\probaZ \space  ln 2 ln div \probaZ \space mul neg 1 \probaZ \space neg add ln 2 ln div 1 \probaZ \space neg add mul neg add}
\def\IC{\entropyU \space \entropyZsU \space add \entropyZ \space neg add #5 neg add}

\psline{->}(0,-0.1)(0,1.2)
\psline{->}(-0.1,0)(1.2,0)
\psline{-}(1,-0.1)(1,1.2)
\rput[u](-0.03,-0.03){$0$}
\rput[l](1.05,1){$1$}
\rput[u](1.05,0.05){$1$}
\rput[u](1.15,-0.05){$p(u_2)$}
\psline(#1,0)(#1,1.2)
\rput[l]{45}(#1,1.24){$ p_0 = #1$}
\psplot[plotpoints=100]{0.001}{0.999}{x ln 2 ln div x mul neg 1 x neg add ln 2 ln div 1 x neg add mul neg add}
\psline[linecolor=black](! 0 \IC)(! 1 \IC)
\psdots[linecolor=black](! #1 \IC)(! 0 \IC)
\rput[r](!-0.02 \IC){$H(U|Z) - C$}

\psline[linecolor=orange, linewidth=3.5pt](0,0)(#4,0)
\psline[linecolor=orange, linewidth=3.5pt](#4,1)(1,1)
\psline[linestyle=dotted,linecolor=orange](#4,0)(#4,1)
\psdots(#1,0)
\uput[u](0.8,1.05){$\textcolor[rgb]{0,0.6,0}{v_1^{\star}}$}
\uput[u](0.36,1.05){$\textcolor[rgb]{0,0,1}{v_2^{\star}}$}
\rput[l]{45}(\qqA,1.24){$\textcolor[rgb]{0,0,1}{p_{01}= \qqA}$}
\psline[linestyle=dotted,linecolor=blue](\qqA,0)(\qqA,1.2)
\psdots[linecolor=blue](\qqA,0)(\qqA,1)(\pAnuB,0)(0.6506,0)
\psdots[linecolor=blue,dotscale=2](\pAnuB,1)
\psline[linestyle=dotted,linecolor=blue](\pAnuB,0)(\pAnuB,1)
\psdots[linecolor=blue,dotscale=2.5, dotstyle=square*](0.6506,1)
\psline[linestyle=dotted,linecolor=blue](0.6506,0)(0.6506,1)
\psdots[linecolor=blue](0.6506,0.9335)(0.9692,0.1983)(\qqA,0.7041)
\psline[linestyle=dashed,linecolor=blue](0.6506,0.9335)(0.9692,0.1983)

\rput[l]{45}(\qqB,1.24){$\textcolor[rgb]{1,0,0}{p_{02}=\qqB}$}
\psline[linestyle=dotted,linecolor=red](\qqB,0)(\qqB,1.2)
\psline[linestyle=dotted,linecolor=red](0.0815,0)(0.0815,0.4073)
\psline[linestyle=dotted,linecolor=red](#4,0)(#4,1)

\rput[l]{-45}(0.0815,-0.045){$p_2=p_2(\nu_1)=0.0815$}
\rput[l]{-45}(\pAnuB,-0.045){$p_3=p_1(\nu_2)=\pAnuB$}
\rput[l]{-45}(#4,-0.045){$p_4 = \gamma=#4$}
\rput[l]{-45}(0.6906,-0.045){$p_1 = 0.6506$}
\psdots[linecolor=red,dotscale=2](#4,1)
\psdots[linecolor=red](\qqB,0)
\psdots[linecolor=red,dotscale=2.5, dotstyle=square*](0.0815,0)
\psline[linestyle=dashed,linecolor=orange](0.0815,0)(#4,1)
\psdots[linecolor=red](0.0815,0.4073)(0.6,0.9710)(\qqB,0.4546)
\psline[linestyle=dashed,linecolor=red](0.0815,0.4073)(0.6,0.9710)
\psline[linecolor=black](\qqA,0.7041)(\qqB,0.4546)

\psline[linestyle=dashed,linecolor=brown](\qqB, 0.0838)(\qqA,1)
\psline[linestyle=dashed,linecolor=brown](#1,0.6337)(1,0.6337)
\psdots[linecolor=brown](\qqB, 0.0838)(\qqA,1)(#1,0.6337)(1,0.6337)
\rput[l](1.02, 0.6337){$\Phi_{\textsf{e}}^{\star} = 0.6337$}

\end{pspicture}
\caption{Optimal encoder's utility depending on the belief $p(u_2)\in [0,1]$ after observing $(W,Z)$,
for parameters $p_0=#1$, $\delta_1=#2$, $\delta_2=#3$ and $\gamma = #4$. The curve represents the binary entropy $H_b(\cdot)$, as a function of $p(u_2)\in [0,1]$.}\label{fig:PlayerOneUtilityEntropy#1_#2_#3_#4_#5}
\end{center}
\end{figure}}

\newcommand{\FigPosteriorsRegion}[5]{

\begin{figure}[!ht]
\begin{center}
\psset{xunit=6cm,yunit=5.5cm}
\begin{pspicture}(0.2,-0.2)(0.8,1.2)
\pscustom[linewidth=0pt,linestyle=none,fillstyle=solid,fillcolor=blue]{
\fileplot[]{DataSide1_#1_#2_#3_#4_#5.dat}
\psline{}(#1,1)(#1,#1)(0,#1)
}
\fileplot[]{DataSide1_#1_#2_#3_#4_#5.dat}
\pscustom[linewidth=0pt,linestyle=none,fillstyle=solid,fillcolor=blue]{
\fileplot[]{DataSide2_#1_#2_#3_#4_#5.dat}
\psline{}(1,#1)(#1,#1)(#1,0)
}
\fileplot[]{DataSide2_#1_#2_#3_#4_#5.dat}
\pscustom[linewidth=0pt,linestyle=none,fillstyle=solid,fillcolor=green]{
\fileplot[]{Data1_#1_#2_#3_#4_#5.dat}
\psline{}(#1,1)(#1,#1)(0,#1)
}
\fileplot[]{Data1_#1_#2_#3_#4_#5.dat}
\pscustom[linewidth=0pt,linestyle=none,fillstyle=solid,fillcolor=green]{
\fileplot[]{Data2_#1_#2_#3_#4_#5.dat}
\psline{}(1,#1)(#1,#1)(#1,0)
}
\fileplot[]{Data2_#1_#2_#3_#4_#5.dat}
\pspolygon*[linecolor=gray](0,0)(#1,0)(#1,#1)(0,#1)
\pspolygon*[linecolor=gray](1,1)(#1,1)(#1,#1)(1,#1)
\rput[l]{-45}(#1,-0.04){$ p_0 = #1$}
\rput[r](-0.05,#1){$ p_0 = #1$}
\psline{-}(#1,0)(#1,1)
\psline{-}(0,#1)(1,#1)
\psline{->}(0,-0.1)(0,1.2)
\psline{->}(-0.1,0)(1.2,0)
\psline{-}(1,-0.1)(1,1.2)
\psline{-}(-0.1,1)(1.2,1)
\rput[u](-0.03,-0.03){$0$}
\rput[u](-0.03,1.03){$1$}
\rput[u](-0.1,1.15){$q_2$}
\rput[u](1.03,-0.03){$1$}
\rput[u](1.25,-0.03){$q_1$}
\psdots(#1,0)(0,#1)
\psline[linestyle=dotted,linecolor=red](0.383,0)(0.383,0.913)(0,0.913)
\psdots[linecolor=red,dotscale=2 2](0.383,0.913)
\rput[l]{-45}(0.383,-0.045){$\textcolor[rgb]{1,0,0}{\nu_3=0.383}$}
\rput[r](-0.03,0.913){$\textcolor[rgb]{1,0,0}{\nu_2=0.913}$}
\end{pspicture}
\caption{Regions of posterior beliefs $(q_1,q_2)$ satisfying the information constraints: $I(U;W|Z) \leq C$ in dark blue, and $I(U;W) \leq C$ in green, for $p_0=#1$, $\delta_1=#2$, $\delta_2=#3$ and channel capacity $C = #5$.}
\label{fig:PosteriorsRegion#1_#2_#3_#4_#5}
\end{center}
\end{figure}}

\ifCLASSINFOpdf
\else
\fi

\begin{document}


\title{Information-Theoretic Limits of Strategic Communication}

%

%
%
%


\author{\IEEEauthorblockN{Ma\"{e}l Le Treust\IEEEauthorrefmark{1} and 
Tristan Tomala \IEEEauthorrefmark{2}}\\
\IEEEauthorblockA{\IEEEauthorrefmark{1}
ETIS UMR 8051, Université Paris Seine, Université Cergy-Pontoise, ENSEA, CNRS,\\
6, avenue du Ponceau, 95014 Cergy-Pontoise CEDEX, FRANCE\\
Email: mael.le-treust@ensea.fr}\\
\thanks{\IEEEauthorrefmark{1} Maël Le Treust gratefully acknowledges the support of the Labex MME-DII (ANR11-LBX-0023-01) and the DIM-RFSI of the Région Île-de-France, under grant EX032965.}
\IEEEauthorblockA{\IEEEauthorrefmark{2}
HEC Paris, GREGHEC UMR 2959\\
1 rue de la Libération, 78351 Jouy-en-Josas CEDEX, FRANCE\\
Email: tomala@hec.fr}
\thanks{\IEEEauthorrefmark{2} Tristan Tomala gratefully acknowledges the support of the HEC foundation and ANR/Investissements d'Avenir under grant ANR-11-IDEX-0003/Labex Ecodec/ANR-11-LABX- 0047.} \thanks{This work was presented in part at the 54th Allerton Conference, Monticello, Illinois, Sept. 2016 \cite{LeTreustTomala(Allerton)16}; the XXVI Colloque Gretsi, Juan-Les-Pins, France, Sept. 2017 \cite{LeTreustTomala(Gretsi)17}; the International Zurich Seminar on Information and Communication, Switzerland, Feb. 2018 \cite{LeTreustTomala(IZS)18}. The authors would like to thank Institute Henri Poincaré (IHP) in Paris, France, for hosting numerous research meetings.}}


\vspace{-0.9cm}

\maketitle

%

\IEEEpeerreviewmaketitle





\vspace{-1.2cm}

\begin{abstract}

In this article, we investigate strategic information transmission over a noisy channel. This problem has been widely investigated in Economics, when the communication channel is perfect. Unlike in Information Theory, both encoder and decoder have distinct objectives and choose their encoding and decoding strategies accordingly. This approach radically differs from the conventional Communication paradigm, which assumes transmitters are of two types: either they have a common goal, or they act as opponent, e.g. jammer, eavesdropper. We formulate a point-to-point source-channel coding problem with state information, in which the encoder and the decoder choose their respective encoding and decoding strategies in order to maximize their long-run utility functions. This strategic coding problem is at the interplay between Wyner-Ziv's scenario and the Bayesian persuasion game of Kamenica-Gentzkow. We characterize a single-letter solution and we relate it to the previous results by using the concavification method. This confirms the benefit of sending encoded data bits even if the decoding process is not supervised, e.g. when the decoder is an autonomous device. Our solution has two interesting features: it might be optimal not to use all channel resources; the informational content impacts the encoding process, since utility functions capture preferences on source symbols.

\end{abstract}


\section{Introduction}\label{sec:Introduction}

What are the limits of information transmission between autonomous devices? The internet of things (IoT) considers pervasive presence in the environment of a variety of devices that are able to interact and coordinate with each other in order to create new applications/services and reach common goals. Unfortunately the goals of wireless devices are not always common, for example adjacent access points in crowded downtown areas, seeking to transmit at the same time, compete for the use of bandwidth. Such situations require new efficient techniques to coordinate the actions of the devices whose objectives are \emph{neither aligned, nor antagonistic}. This question differs from the classical paradigm in Information Theory which assumes that devices are of two types: transmitters pursue the common goal of transferring information, while opponents try to mitigate the communication (e.g. the jammer corrupts the information, the eavesdropper infers it, the warden detects the covert transmission). In this work, we seek to provide an information-theoretic look at the intrinsic limits of the strategic communication between interacting autonomous devices having non-aligned objectives.


The problem of ``strategic information transmission'' has been well studied in the Economics literature since the seminal paper by Crawford-Sobel \cite{CrawfordSobel1982StrategicInformation}. In their model, a better-informed sender transmits a signal to a receiver, who takes an action which impacts both sender/receiver's utility functions. The problem consists in determining the  \emph{optimal information disclosure} policy given that the receiver' best-reply action will impact the sender's utility, see \cite{Forges94} for a survey. In \cite{KamenicaGentzkow11}, Kamenica-Gentzkow define the ``Bayesian persuasion game'' where the sender \emph{commits} to an information disclosure policy before the game starts. This subtle change of rules of the game induces a very different equilibrium solution related to \emph{Stackelberg equilibrium} \cite{stackelberg-book-1934}, instead of \emph{Nash equilibrium} \cite{Nash51}. This problem was later referred to as ``information design'' in \cite{BergemannMorris16}, \cite{Taneva16}, \cite{BergemannMorris17}. In most of the articles in the Economics literature, the transmission between the sender and the receiver is noise-free; except the following ones \cite{TsakasTsakas2017}, \cite{Blume}, \cite{HernandezVonStengel14} where the noisy transmission is investigated in a finite block-length regime. Noise-free transmission does not require the use of block coding techniques for compressing the information and mitigating the noise effects. Interestingly, Shannon's mutual information is widely accepted as a \emph{cost of information} for the problem of ``rational inattention'' in \cite{Sims03} and for the problem of ``costly persuasion'' in \cite{GentzkowKamenica14}; in both cases there is no explicit reference to the coding problem.

Entropy and mutual information appear endogenously in repeated games with finite automata and bounded recall \cite{NeymanOkada99}, \cite{NeymanOkada00}, \cite{NeymanOkada09}, with private observation \cite{GossnerVieille02}, or with imperfect monitoring \cite{GossnerTomala06}, \cite{GossnerTomala07}, \cite{GossnerLarakiTomala09}. In \cite{GossnerHernandezNeyman06}, the authors investigate a sender-receiver game with common interests by formulating a coding problem and by using tools from Information Theory. In their model, the sender observes an infinite source of information and communicates with the receiver through a perfect channel of fixed alphabet. This result was later refined by Cuff in  \cite{Cuff(ImplicitCoordination)11} and referred to as the  ``coordination problem'' in several articles  \cite{CuffPermuterCover10}, \cite{CuffSchieler11}, \cite{LeTreust(EmpiricalCoordination)17}, \cite{CerviaLuzziLeTreustBloch(IT)18}.

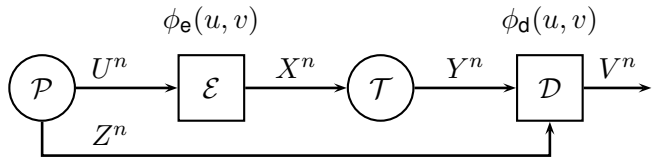
\begin{figure}[!ht]
\begin{center}
\psset{xunit=0.9cm,yunit=0.9cm}
\begin{pspicture}(0,-1.1)(8.5,1.7)
\pscircle(0,0.5){0.45}
\psframe(2,0)(3,1)
\pscircle(5,0.5){0.45}
\psframe(7,0)(8,1)
\psline[linewidth=1pt]{->}(0.5,0.5)(2,0.5)
\psline[linewidth=1pt]{->}(3,0.5)(4.5,0.5)
\psline[linewidth=1pt]{->}(5.5,0.5)(7,0.5)
\psline[linewidth=1pt]{->}(8,0.5)(9,0.5)
\psline[linewidth=1pt]{->}(0,0)(0,-0.5)(7.5,-0.5)(7.5,0)
\rput[u](1,-0.2){$Z^{n}$}
\rput[u](1,0.8){$U^{n}$}
\rput[u](3.75,0.8){$X^n$}
\rput[u](6.25,0.8){$Y^n$}
\rput[u](8.5,0.8){$V^n$}
\rput(0,0.5){$\PP$}
\rput(5,0.5){$\mc{T}$}
\rput(2.5,0.5){$\EE$}
\rput(7.5,0.5){$\DD$}
\rput(2.5,1.5){$\phi_{\textsf{e}}(u,v)$}
\rput(7.5,1.5){$\phi_{\textsf{d}}(u,v)$}
\end{pspicture}
\caption{The information source is i.i.d. $\PP(u,z)$ and the channel $\mc{T}(y|x)$ is memoryless. The encoder $\EE$ and the decoder $\DD$ are endowed with distinct utility functions $\phi_{\textsf{e}}(u,v) \in \R$ and $\phi_{\textsf{d}}(u,v) \in \R $. }
\label{fig:StrategicEmpiricalCoordination}
\end{center}
\end{figure}

In the literature of Information Theory, the Bayesian persuasion game is investigated in   \cite{AkyolLangbortBasarIEEE17} for Gaussian source and channel with Crawford-Sobel's quadratic cost functions. The authors compute the linear equilibrium's encoding and decoding strategies and relate their results to the literature on ``decentralized stochastic control''. In \cite{SaritasYukselGezici17}, \cite{SaritasYukselGezici(TAC)17}, the authors extend the model of Crawford-Sobel to multidimentional sources and noisy channels. They determine whether the optimal encoding policies are linear or based on  quantification. Sender-receiver games are also investigated in \cite{AkyolLangbortBasar(GSIP)17}, \cite{FarokhiTeixeiraLangbort17} for the problem of ``strategic estimation'' involving self-interested sensors; and in \cite{MarecekShortenYu15} for the ``strategic congestion control'' problem. In \cite{DughmiXu16}, \cite{Dughmi17}, \cite{DughmiKempeQiang16}, the authors investigate the computational aspects of the Bayesian persuasion game when the signals are noisy. In \cite{BerryTse(ShannonMetNash)11}, \cite{PerlazaTandonPoorHan15} the interference channel coding problem is formulated  as a game in which the users, i.e. the pairs of encoder/decoder, are allowed to use \emph{any} encoding/decoding strategy. The authors compute the set of Nash equilibria for linear deterministic and Gaussian channels.

The non-aligned devices' objectives are captured by utility functions, defined in a similar way as the distortion function for lossy source coding. Coding for several distortion measures is investigated for ``multiple descriptions coding'' in  \cite{GamalCover82}, for the lossy version of ``Steinberg's common reconstruction'' problem in \cite{LapidothMalarWigger14}, for an alternative measure of ``secrecy'' in \cite{Yamamoto88}, \cite{Yamamoto97}, \cite{SchielerCuff(RateDistortion14)}, \cite{SchielerCuff(Henchman)16}. In  \cite{Lapidoth97}, Lapidoth investigates the ``mismatch source coding problem'' in which the sender is constrained by Nature, to encode according to a different distortion than the decoder's one. The coding problems in these works consider two types of transmitters: either they pursue a common goal, or they behave as opponents.

The problem of information transmission has therefore been addressed from two complementary point of views: the economists consider noiseless environment and non-aligned objectives whereas the information theorists consider noisy environment and the objectives are either aligned or opposed. In this work, we propose a novel framework in order to investigate both aspects simultaneously, by considering noisy environment and non-aligned transmitters' objectives. We formulate the \emph{strategic coding problem} by considering a joint source-channel scenario in which the decoder observes a state information à la Wyner-Ziv \cite{wyner-it-1976}. This model generalizes the framework we introduced in our previous work, in \cite{LeTreustTomala17}. The encoder and the decoder are endowed with distinct utility functions and they play a multi-dimensional version of the Bayesian persuasion game of Kamenica-Gentkow in \cite{KamenicaGentzkow11}. We point out two essential features of the strategic coding problem:
\begin{itemize}
\item[$\bullet$] Each source symbol has a different impact on encoder/decoder's utility functions; so it's optimal to encode each symbol differently.
\item[$\bullet$] In the noiseless version of the Bayesian persuasion game in \cite{KamenicaGentzkow11}, the optimal information disclosure policy requires a fixed amount of information bits. When the channel capacity is larger than this amount, it is optimal not to use all the channel resource.
\end{itemize}
The contributions of this article are as follows:
\begin{itemize}
\item[$\bullet$] We characterize the single-letter solution of the strategic coding problem and we relate it to Wyner-Ziv's rate-distortion function \cite{wyner-it-1976} and to the separation result by Merhav-Shamai \cite{MerhavShamai03}. This characterization is based on the set of probability distributions that are achievable for the related coordination problem, under investigation in \cite[Theorem IV.2]{LeTreust(ISIT-TwoSided)15}.
\item[$\bullet$] We reformulate our single-letter solution in terms of a concavification as in Kamenica-Gentkow \cite{KamenicaGentzkow11}, taking into account the information constraint imposed by the noisy channel. This provides an alternative point of view on the problem: given an encoding strategy, the decoder computes its posterior beliefs and chooses a best-reply action accordingly. Knowing this in advance, the encoder discloses the information optimally such as to induce the posterior beliefs corresponding to its optimal actions.
\item[$\bullet$] We reformulate our solution in terms of a concavification of a Lagrangian, so as to relate with the \emph{cost of information} considered for rational inattention in \cite{Sims03} and for the costly persuasion in \cite{GentzkowKamenica14}. We also provide a bi-variate concavification where the information constraint is integrated along an additional dimension.
\item[$\bullet$] One technical novelty is the characterization of the posterior beliefs induced by Wyner-Ziv's coding. This confirms the benefit of sending encoded data bits to a autonomous decoder, even  if the decoding process is not controlled. In fact, we prove that Wyner-Ziv's coding reveals the \emph{exact} amount of information needed, no less no more. This property also holds for Shannon's Lossy source coding, as demonstrated in \cite{LeTreustTomala17}. 
\item[$\bullet$] When the channel is perfect and has a large input alphabet, our coding problem is equivalent to several i.i.d. copies of the one-shot problem, whose optimal solution is given by our characterization without the information constraint. This noise-free setting is related to the problem of ``persuasion with heterogeneous beliefs'' under investigation in \cite{AlonsoCamara(JET)2016} and \cite{LaclauRenou17}.
\item[$\bullet$] We illustrate our results by considering an example with binary source, states and decoder's actions. We explain the concavification method and we analyse the impact of the channel noise on the set of achievable posterior beliefs.
\item[$\bullet$] Surprisingly, the decoder's state information has two opposite effects on the optimal encoder's utility: it enlarges the set of decoder's posterior beliefs, so it may increase encoder's utility; it reveals partial information to the decoder, so it forces some decoder's best-reply actions that might be sub-optimal for the encoder, hence it may decrease encoder's utility.
\end{itemize}

The strategic coding problem is formulated in Sec. \ref{sec:StrategicCoding}. The encoding and decoding strategies and the utility functions are defined in Sec. \ref{sec:StrategyUtility}. The persuasion game is introduced in Sec. \ref{sec:PersuasionGame}. Our coding result and the four different characterizations are stated in Sec. \ref{sec:Characterization}. The first one is a linear program under an information constraint, formulated in Sec. \ref{sec:LinearProgram}. The main Theorem is stated in Sec.\ref{sec:MainResult}. In Sec. \ref{sec:Concavification}, we reformulate our solution in terms of three different concavifications, related to the previous results. Sec. \ref{sec:ExampleBinary} provides an example based on a binary source, binary states and binary decoder's actions. The solution is provided without information constraint in Sec. \ref{sec:ExampleConcavificationNo} and with information constraint in Sec. \ref{sec:concavificationIC}. The proofs are stated in App \ref{sec:ProofConcavification} - \ref{sec:ProofLemmaPosteriors}.


\section{Strategic Coding Problem}\label{sec:StrategicCoding}

\subsection{Coding Strategies and Utility Functions}\label{sec:StrategyUtility}

We consider the i.i.d. distribution of information source/state $\PP(u,z)$ and the memoryless channel  distribution $\mc{T}(y|x)$ depicted in Fig. \ref{fig:StrategicEmpiricalCoordination}. Uppercase letters $U$ denote the random variables, lowercase letters $u$ denote the realizations and calligraphic fonts $\mc{U}$ denote the alphabets. Notations $U^n$, $X^n$, $Y^n$, $Z^n$, $V^n$ stand for sequences of random variables of information source $u^n=(u_1,\ldots,u_n)\in\mc{U}^n$, decoder's state information $z^n\in\mc{Z}^n$, channel inputs $x^n\in\mc{X}^n$, channel outputs $y^n\in\mc{Y}^n$ and decoder's actions $v^n\in\mc{V}^n$, respectively. The sets $\mc{U}$, $\mc{Z}$, $\mc{X}$, $\mc{Y}$, $\mc{V}$ have finite cardinality and the notation $\Delta(\mc{X})$ stands for the set of  probability distributions over $\mc{X}$, i.e. the probability simplex. With a slight abuse of notation, we denote by $\QQ(x)\in \Delta(\mc{X})$ the probability \emph{distribution}, as in \cite[pp. 14]{cover-book-2006}. For example the joint probability distribution $\QQ(x,v)   \in \Delta(\mc{X}\times \mc{V})$ decomposes as: $\QQ(x,v)  =\QQ(v) \times \QQ(x|v) =\QQ(x) \times \QQ(v|x) $. The distance between two probability distributions $\QQ(x)$ and $\PP(x)$ is based on $L^1$ norm, denoted by: $||\QQ - \PP||_{1}= \sum_{x\in\mc{X}} |\QQ(x) - \PP(x)|$. Notation $U  -\!\!\!\!\minuso\!\!\!\!-X    -\!\!\!\!\minuso\!\!\!\!-  Y$ stands for the Markov chain property corresponding to $\PP(y|x,u) = \PP(y|x)$, for all $(u,x,y)$. The encoder and the decoder and denoted by $\EE$ and $\DD$.

\begin{definition}[Encoding and Decoding Strategies]\label{def:Code}$\;$\\
$\bullet$ The encoder $\EE$ chooses an encoding strategy $\sigma$ and the decoder $\DD$ chooses a decoding strategy $\tau$, defined by:
\begin{align}
\sigma& : \mc{U}^{n} \longrightarrow \Delta(\mc{X}^n)  ,\label{eq:EncodingFunction}\\
\tau& : \mc{Y}^n \times   \mc{Z}^n  \longrightarrow  \Delta( \mc{V}^n)  . \label{eq:DecodingFunction}
\end{align}
Both strategies $(\sigma,\tau)$ are stochastic.\\
$\bullet$ The strategies $(\sigma,\tau)$ induces a joint probability distribution $\PP_{\sigma,\tau} \in\Delta(\mc{U}^{n} \times\mc{Z}^{n} \times\mc{X}^{n}\times\mc{Y}^{n}  \times\mc{V}^{n} )$ over the $n$-sequences of symbols, defined by:
\begin{align}
 &\prod_{i=1}^n\PP\big(u_i,z_i \big) \times\sigma\big(x^n\big| u^n \big)
 \times \prod_{i=1}^n \mc{T}\big(y_i \big|x_i\big) \times \tau\big(v^n \big| y^n ,z^n\big).
\end{align}
\end{definition}

The encoding and decoding strategies $(\sigma,\tau)$ correspond to the problem of {joint source-channel coding with decoder's state information} studied in \cite{MerhavShamai03}, based on Wyner-Ziv's setting in \cite{wyner-it-1976}. Unlike these previous works, we investigate the case where the encoder and the decoder are autonomous decision-makers, who choose their own encoding $\sigma$ and decoding $\tau$ strategies. In this work, the encoder and the decoder have non-aligned objectives captured by distincts utility functions, depending on the source symbol $U$, the decoder's state information $Z$ and on decoder's action $V$.

\begin{definition}[Utility Functions]\label{def:Utilities} 
$\bullet$ The single-letter utility functions of $\EE$ and $\DD$ are defined by:
\begin{align}
\phi_{\textsf{e}} : \mc{U} \times \mc{Z} \times \mc{V} \longrightarrow \R,\\
\phi_{\textsf{d}} : \mc{U} \times \mc{Z} \times \mc{V} \longrightarrow \R.
\end{align}
$\bullet$ The long-run utility functions $\Phi_{\textsf{e}}^n(\sigma,\tau)$ and $\Phi_{\textsf{d}}^n(\sigma,\tau)$ are evaluated with respect to the probability distribution $\PP_{\sigma,\tau}$ induced by the strategies $(\sigma,\tau)$:
\begin{align}
\Phi_{\textsf{e}}^n(\sigma,\tau) =& \E_{\sigma,\tau} \Bigg[ \frac{1}{n} \sum_{i=1}^n \phi_{\textsf{e}}(U_i,Z_i,V_i) \Bigg] \nonumber\\
=& \sum_{u^n,z^n,v^n}\PP_{\sigma,\tau}\big(u^n,z^n,v^n \big) \cdot  \Bigg[  \frac{1}{n} \sum_{i=1}^n \phi_{\textsf{e}}(u_i,z_i,v_i) \Bigg],\\
\Phi_{\textsf{d}}^n(\sigma,\tau) =& \sum_{u^n,z^n,v^n}\PP_{\sigma,\tau}\big(u^n,z^n,v^n \big) \cdot  \Bigg[  \frac{1}{n} \sum_{i=1}^n \phi_{\textsf{d}}(u_i,z_i,v_i) \Bigg].
\end{align}
\end{definition}

\subsection{Bayesian Persuasion Game}\label{sec:PersuasionGame}

We investigate the strategic communication between autonomous devices who choose the encoding $\sigma$ and decoding $\tau$ strategies in order to \emph{maximize} their own long-run utility functions $\Phi_{\textsf{e}}^n(\sigma,\tau)$ and $\Phi_{\textsf{d}}^n(\sigma,\tau)$. We assume that the encoding strategy $\sigma$ is designed and observed by the decoder in advance, i.e. before the transmission starts; then the decoder \emph{is free to choose any decoding strategy} $\tau$. This framework corresponds to the \emph{Bayesian persuasion game} \cite{KamenicaGentzkow11}, in which the encoder commits to its strategy $\sigma$ and announces it to the decoder, who chooses strategy $\tau$ accordingly. We assume that the strategic communication takes place as follows:
\begin{itemize}
\item[$\bullet$] The encoder $\EE$ chooses and announces an encoding strategy $\sigma$ to the decoder $\DD$. 
\item[$\bullet$] The sequences $(U^n, Z^n,X^n,Y^n)$ are drawn according to the probability distribution: $\prod_{i=1}^n\PP(u_i,z_i) \times \sigma(x^n|u^n)\times \prod_{i=1}^n\mc{T}(y_i|x_i)$.
\item[$\bullet$] The decoder $\DD$ knows $\sigma$, observes the sequences of symbols $(Y^n,Z^n)$  and draws a sequence of actions $V^n$ according to $\tau(v^n|y^n,z^n)$.
\end{itemize}
By knowing $\sigma$ in advance, the decoder $\DD$ can compute the set of best-reply decoding strategies. 
\begin{definition}[Decoder's Best-Replies]\label{def:BestReply} 
For any encoding strategy $\sigma$, the set of best-replies decoding strategies $ \textsf{BR}_{\textsf{d}}(\sigma)$, is defined by:
\begin{align}
 \textsf{BR}_{\textsf{d}}(\sigma) =& \bigg\{\tau ,\text{ s.t. } \;  \Phi_{\textsf{d}}^n(\sigma, \tau) \geq \Phi_{\textsf{d}}^n(\sigma, \widetilde{\tau})  , \; \forall \widetilde{\tau} \neq \tau \bigg\}.
\end{align}
\end{definition}
In case there are several best-reply strategies, we assume that the decoder chooses the one that minimizes encoder's utility: $\min_{\tau \in \textsf{BR}_{\textsf{d}}(\sigma)} \Phi_{\textsf{e}}^n(\sigma, \tau)$, so that encoder's utility is robust to the exact specification of decoder's strategy. 

The coding problem under investigation consists in maximizing the encoder's long-run utility:
\begin{align}
\sup_{\sigma}\min_{\tau \in \textsf{BR}_{\textsf{d}}(\sigma)} \Phi_{\textsf{e}}^n(\sigma, \tau). \label{eq:GameProblem}
\end{align}
Problem \eqref{eq:GameProblem} raises the following interesting question: is it optimal for an autonomous decoder to \emph{extract} the encoded information? We provide a positive answer to this question 
by showing that the actions induced by Wyner-Ziv's decoding $\tau^{\textsf{wz}}$ coincide with those induced by any best-reply $\tau^{\textsf{BR}}\in\textsf{BR}_{\textsf{d}}(\sigma^{\textsf{wz}})$ to Wyner-Ziv's encoding $\sigma^{\textsf{wz}}$, for a large fraction of stages. We characterize a single-letter solution to \eqref{eq:GameProblem}, by refining Wyner-Ziv's result for source coding with decoder's state information \cite{wyner-it-1976}. 


\begin{remark}[Stackelberg v.s. Nash Equilibrium]
The optimization problem of \eqref{eq:GameProblem} is a \emph{Bayesian persuasion game} \cite{KamenicaGentzkow11}, \cite{GentzkowKamenica14} also referred to as \emph{Information Design} problem \cite{BergemannMorris16}, \cite{Taneva16}, \cite{BergemannMorris17}. This corresponds to a Stackelberg equilibrium \cite{stackelberg-book-1934} in which the encoder is the leader and the decoder is the follower, unlike the Nash equilibrium  \cite{Nash51}  in which the two devices choose their strategy simultaneously.
\end{remark}

\begin{remark}[Equal Utility Functions]
When assuming that the encoder and decoder have a common objective, i.e. have equal utility function $\phi_{\textsf{e}} =\phi_{\textsf{d}}$, our problem boils down to the classical approach of Wyner-Ziv \cite{wyner-it-1976} and Merhav-Shamai \cite{MerhavShamai03}, in which both strategies $(\sigma,\tau)$ are chosen \emph{jointly}, in order to maximize the utility function:
\begin{align}
\sup_{\sigma}\min_{\tau \in \textsf{BR}_{\textsf{d}}(\sigma)} \Phi_{\textsf{e}}^n(\sigma, \tau) = 
\max_{(\sigma,\tau)} \Phi_{\textsf{e}}^n(\sigma, \tau), \label{eq:GameProblemEqual}
\end{align}
or to minimize a distortion function. 
\end{remark}

%
%
%


\section{Characterizations}\label{sec:Characterization}

\subsection{Linear Program with Information Constraint}\label{sec:LinearProgram}

Before stating our main result, we define the encoder's optimal utility  $\Phi_{\textsf{e}}^{\star}$.\begin{definition}[Target Distributions]\label{def:Characterization} 
We consider an auxiliary random variable $W\in \mc{W}$ with $|\mc{W}| = \min\big(|\mc{U}|+1, |\mc{V}|^{|\mc{Z}|}\big)$. The set $\Q_0$ of target probability distributions is defined by:
\begin{align}
\Q_0 =& \bigg\{  \PP(u,z) \times \QQ(w | u) ,  \quad \text{s.t.}, 
\quad\;\;\max_{\PP(x)} I( X; Y )  -   I( U ;W |Z )   \geq 0  \bigg\}.\label{eq:SetQ0}
\end{align} 
We define the set $\Q_2\big(\QQ(u,z,w)\big)$ of single-letter best-replies of the decoder:
\begin{align}
\Q_2\big(\QQ(u,z,w)\big) =& \argmax_{\QQ(v |z,w)}\E_{\QQ(u,z,w)\atop \times \QQ(v |z,w)  } \bigg[ \phi_{\textsf{d}}(U,Z,V) \bigg].\label{eq:SetQ2}
\end{align} 
The encoder's optimal utility $\Phi_{\textsf{e}}^{\star} $ is given by:
\begin{align}
\Phi_{\textsf{e}}^{\star} =&  \sup_{\QQ(u,z,w) \in \Q_0} \min_{\QQ(v |z,w) \in  \atop \Q_2(\QQ(u,z,w))} \E_{\QQ(u,z,w) \atop \times \QQ(v |z,w)} \bigg[\phi_{\textsf{e}}(U,Z,V)\bigg].\label{eq:PhiOptimalZ}
\end{align}
\end{definition}
We discuss the above definitions.\\
$\bullet$ The information constraint  \eqref{eq:SetQ0} of the set $\Q_0$ involves the channel capacity $\max_{\PP(x)} I( X; Y )$ and the Wyner-Ziv's information rate $ I( U ;W |Z ) = I( U ;W ) - I(W;Z )$, stated in \cite{wyner-it-1976}. It corresponds to the separation result by Shannon \cite{shannon-bell-1948}, extended to the Wyner-Ziv setting by Merhav-Shamai in \cite{MerhavShamai03}. \\
$\bullet$ For the clarity of the presentation, the set $\Q_2\big(\QQ(u,z,w)\big)$ contains stochastic functions $\QQ(v |z,w)$, even if for the linear problem \eqref{eq:SetQ2}, some optimal $\QQ(v |z,w)$ are deterministic. If there are several optimal $\QQ(v |z,w)$, we assume the decoder chooses the one that minimize encoder's utility: $\min_{\QQ(v |z,w) \in  \atop \Q_2(\QQ(u,z,w))} \E \big[\phi_{\textsf{e}}(U,Z,V)\big]$, so that encoder's utility is robust to the exact specification of $\QQ(v |z,w)$.\\
$\bullet$ The supremum over $\QQ(u,z,w) \in \Q_0$ is not a maximum since the  function $\min_{\QQ(v |z,w) \in  \atop \Q_2(\QQ(u,z,w))} \E \big[\phi_{\textsf{e}}(U,Z,V)\big]$ is not continuous with respect to $\QQ(u,z,w)$.  \\
$\bullet$ In \cite[Theorem IV.2]{LeTreust(ISIT-TwoSided)15}, the author shows that the sets $\Q_0$ and $\Q_2$  correspond to the target probability distributions $\QQ(u,z,w)\times \QQ(v|z,w)$ that are achievable for the problem of \emph{empirical coordination}, see also \cite{CuffPermuterCover10}, \cite{LeTreust(EmpiricalCoordination)17}. As noticed in \cite{LeTreustBloch(ISIT)16} and \cite{LeTreustBloch(StateLeakageCoordination)18}, the tool of Empirical Coordination allows us to characterize the ``core of the decoder's knowledge'', that captures what the decoder can infer about all the random variables of the problem.\\
$\bullet$ The value $\Phi_{\textsf{e}}^{\star} $ corresponds to the Stackelberg equilibrium of an auxiliary one-shot  game in which the decoder chooses $\QQ(v |z,w)$, knowing in advance that the encoder has chosen $ \QQ(w | u)\in \Q_0$ and the utility functions are: $\E\big[\phi_{\textsf{e}}(U,Z,V)\big]$ and $\E\big[\phi_{\textsf{d}}(U,Z,V)\big]$.

\begin{remark}[Equal Utility Functions]
When assuming that the decoder's utility function is equal to the encoder's utility function: $\phi_{\textsf{d}}(u,z,v) = \phi_{\textsf{e}}(u,z,v)$, then the set $\Q_2\big(\QQ(u,z,w)\big)$ is equal to $ \argmax_{\QQ(v |z,w)}\E\big[ \phi_{\textsf{e}}(U,Z,V) \big]$. Thus, we have:
\begin{align}
 \min_{\QQ(v |z,w) \in  \atop \Q_2(\QQ(u,z,w))} \E_{\QQ(u,z,w) \atop \times \QQ(v |z,w)} \bigg[\phi_{\textsf{e}}(U,Z,V)\bigg] 
 =&  \max_{\QQ(v |z,w)} \E_{\QQ(u,z,w) \atop \times \QQ(v |z,w)} \bigg[\phi_{\textsf{e}}(U,Z,V)\bigg].
\end{align} 
Hence, the encoder's optimal utility $\Phi_{\textsf{e}}^{\star} $ is equal to:
\begin{align}
\Phi_{\textsf{e}}^{\star} =&  \sup_{\QQ(u,z,w) \in \Q_0} \min_{\QQ(v |z,w) \in  \atop \Q_2(\QQ(u,z,w))} \E_{\QQ(u,z,w) \atop \times \QQ(v |z,w)} \bigg[\phi_{\textsf{e}}(U,Z,V)\bigg]\\
=&  \sup_{\QQ(u,z,w) \in \Q_0} \max_{\QQ(v |z,w)} \E_{\QQ(u,z,w) \atop \times \QQ(v |z,w)} \bigg[\phi_{\textsf{e}}(U,Z,V)\bigg]\label{eq:PhiOptimalZr0}\\
=&  \max_{\QQ(u,z,w) \in \Q_0,\atop \QQ(v |z,w)}\E_{\QQ(u,z,w) \atop \times \QQ(v |z,w)} \bigg[\phi_{\textsf{e}}(U,Z,V)\bigg].\label{eq:PhiOptimalZr}
\end{align}
The supremum in \eqref{eq:PhiOptimalZr0} is replaced by a maximum in \eqref{eq:PhiOptimalZr} due to the compacity of $\Q_0$ and the continuity of function $\max_{\QQ(v |z,w)} \E\bigg[\phi_{\textsf{e}}(U,Z,V)\bigg]$ with respect to $\QQ(u,z,w)$. 

If we consider that the utility function is equal to minus the distortion function: $ \phi_{\textsf{e}}(u,z,v)= - d(u,v)$ as in \cite[Definition 1]{MerhavShamai03}, then we recover the \emph{distortion-rate} function corresponding to \cite[Theorem 1]{MerhavShamai03}: 
\begin{align}
\Phi_{\textsf{e}}^{\star} =& -  \min_{\QQ(u,z,w) \in \Q_0,\atop \QQ(v |z,w)}\E_{\QQ(u,z,w) \atop \times \QQ(v |z,w)} \bigg[d(U,V)\bigg].
\end{align}
\end{remark}


\subsection{Main Result}\label{sec:MainResult}

We introduce the notation $\N^{\star}=\N\setminus\{0\}$ and we characterize the encoder's long-run optimal utility \eqref{eq:GameProblem} by using $\Phi_{\textsf{e}}^{\star} $.

\begin{theorem}[Main Result]\label{theo:MaxMinStackelberg}
The long-run optimal utility of the encoder satisfies:
\begin{align}
&\forall \varepsilon>0,   \; \exists \bar{n}\in \N^{\star}, \;\forall n\geq \bar{n}, \qquad \sup_{\sigma} \min_{\tau \in \textsf{BR}_{\textsf{d}}(\sigma)} \Phi_{\textsf{e}}^n(\sigma, \tau)  \geq  \Phi_{\textsf{e}}^{\star}  -  \varepsilon,\label{eq:Achievability}\\
&\forall n \in \N, \qquad  \sup_{\sigma}\min_{\tau \in \textsf{BR}_{\textsf{d}}(\sigma)} \Phi_{\textsf{e}}^n(\sigma, \tau)  \leq  \Phi_{\textsf{e}}^{\star}.\label{eq:Converse}
\end{align}
\end{theorem}
The proofs of the achievability \eqref{eq:Achievability} and the converse \eqref{eq:Converse} results are given in App. \ref{sec:AchievabilityProof} and \ref{sec:ConverseProof}. When removing decoder's state information $\mc{Z} = \emptyset$, we recover our previous result in \cite[Theorem 4.3]{LeTreustTomala17}. As a consequence, Theorem \ref{theo:MaxMinStackelberg} characterizes the limit behaviour of long-run optimal utility of the encoder:
\begin{align}
\lim_{n \to +\infty} \sup_{\sigma}\min_{\tau \in \textsf{BR}_{\textsf{d}}(\sigma)} \Phi_{\textsf{e}}^n(\sigma, \tau)  =  \Phi_{\textsf{e}}^{\star}.\label{eq:limit}
\end{align}

%


\subsection{Concavification}\label{sec:Concavification}

The concavification of a function $f$ is the smallest concave function $\cav f : \mc{X} \rightarrow \R \cup \{- \infty\}$ that majorizes $f$ on $X$. In this section, we reformulate the encoder's optimal utility $\Phi_{\textsf{e}}^{\star}$ in terms of a  concavification, similarly to \cite[Corollary 1]{KamenicaGentzkow11} and \cite[Definition 4.2]{LeTreustTomala17}. This alternative approach simplifies the optimization problem in \eqref{eq:PhiOptimalZ}, by plugging the decoder's posterior beliefs and best-reply actions into the encoder's utility function. This also provides a nice interpretation: the goal of the strategic communication is to \emph{control the posterior beliefs} of the decoder knowing it will take a best-reply action afterwards.

Before the transmission, the decoder holds a prior belief corresponding to the source's statistics $\PP(u)\in\Delta(\mc{U})$. After observing the pair of symbols $(w,z)\in\mc{W} \times \mc{Z}$, the decoder updates its posterior belief $\QQ(\cdot|z,w)\in \Delta(\mc{U})$ according to Bayes rule: $\QQ(u|z,w) = \frac{\PP(u,z)\QQ(w|u)}{\sum_{u'}\PP(u',z)\QQ(w|u')}$, for all $(u,z,w)\in\mc{U} \times\mc{W} \times \mc{Z}$. 
\begin{definition}[Best-Reply Action]\label{def:BestReply}
For each symbol $z\in \mc{Z}$ and belief $p\in\Delta(\mc{U})$, the decoder chooses the best-reply action $v^{\star}(z,p)$ that belongs to the set $\mc{V}^{\star}(z,p)$, defined by:
\begin{align}
\mc{V}^{\star}(z,p) = & \argmin_{v \in \argmax \E_{p}\big[\phi_{\textsf{d}}(U,z,v)\big] }\E_{p}\bigg[\phi_{\textsf{e}}(U,z,v)\bigg].\label{eq:Vstar}
\end{align}
\end{definition}
If several actions are best-replies to symbol $z\in \mc{Z}$ and belief $p\in\Delta(\mc{U})$, the decoder chooses one of the worst action for encoder's utility. This is a reformulation of the minimum in \eqref{eq:PhiOptimalZ}.
\begin{definition}[Robust Utility Function]\label{def:RobustUtility}
For each symbol $z\in \mc{Z}$ and belief $p\in\Delta(\mc{U})$, the encoder's \emph{robust} utility function is defined by:
\begin{align}
\psi_{\textsf{e}}(z,p) &=   \E_p\big[ \phi_{\textsf{e}}(U,z,v^{\star}(z,p))\big].\label{eq:functionPsi}
\end{align}
\end{definition}
\begin{definition}[Average Utility and Average Entropy]\label{def:AverageUtility}
For each belief $p\in\Delta(\mc{U})$, we define the average encoder's utility function $\Psi_{\textsf{e}}(p)$ and average entropy function $h(p)$:
\begin{align}
\Psi_{\textsf{e}}(p) =& \sum_{u,z} p(u) \cdot \PP(z|u)\cdot  \psi_{\textsf{e}}\bigg(z,\frac{p(u) \cdot  \PP(z|u)}{\sum_{u'}p(u') \cdot  \PP(z|u')}\bigg),\label{eq:FunctionPsip}\\
h(p) =&    \sum_{u,z} p(u) \cdot \PP(z|u)\cdot \log_2 \frac{\sum_{u'}p(u') \cdot  \PP(z|u')}{p(u) \cdot  \PP(z|u)}.\label{eq:FunctionHp}
\end{align}
The conditional probability distribution $\PP(z|u)$ is given by the information source.
\end{definition}

\begin{lemma}[Concavity]\label{lemma:Concavity}
The average entropy $h(p)$ is concave in $p \in \Delta(\mc{U})$.
\end{lemma}
\begin{proof}[Lemma \ref{lemma:Concavity}]
The average entropy $h(p)$ in \eqref{eq:FunctionHp} is a reformulation of the conditional entropy $H(U|Z)$ as a function of the probability distribution $p\in \Delta(\mc{U})$, for a fixed $\PP(z|u)$. The mutual information $I(U;Z)$ is convex in $p\in \Delta(\mc{U})$ (see \cite[pp. 23]{ElGammalKim(book)11}),  and the entropy $H(U)$ is concave in  $p\in \Delta(\mc{U})$. Hence the conditional entropy $h(p) = H(U|Z)= H(U) - I(U;Z)$ is concave in $p\in \Delta(\mc{U})$.
\end{proof}
\begin{theorem}[Concavification]\label{theo:Concavification}
The solution $\Phi_{\textsf{e}}^{\star} $ of \eqref{eq:PhiOptimalZ} is the concavification of $ \Psi_{\textsf{e}}(p)$ evaluated at the prior distribution $\PP(u)$, under an information constraint:
\begin{align}
\Phi_{\textsf{e}}^{\star}  =&  \sup\bigg\{ \sum_{w} \lambda_w \cdot \Psi_{\textsf{e}}(p_w)   \quad \text{ s.t. } \quad \sum_{w} \lambda_w \cdot p_w  = \PP(u) \in \Delta(\mc{U}),\nonumber\\
&\qquad\qquad\qquad\qquad\qquad \text{ and }\quad \sum_{w} \lambda_w \cdot h(p_w)  \geq H(U|Z)  - \max_{\PP(x)}I(X;Y) \bigg\},\label{eq:SplittingFormulation}
\end{align}
where the supremum is taken over $\lambda_w\in [0,1]$ summing up to 1 and $p_w\in \Delta(\mc{U})$, for each $w\in \mc{W}$ with $|\mc{W}| = \min\big(|\mc{U}|+1, |\mc{V}|^{|\mc{Z}|}\big)$.
\end{theorem}
The proof of Theorem \ref{theo:Concavification} is stated in App. \ref{sec:ProofConcavification}. The ``splitting Lemma'' by Aumann and Maschler \cite{AM95}, also called ``Bayes plausibility'' in \cite{KamenicaGentzkow11}, ensures that there is a one-to-one correspondance between the conditional distribution $\QQ(w|u) =  \frac{\lambda_w \cdot p_w(u)}{\PP(u)}$ and the parameters $(\lambda_w,p_w)_{w\in \mc{W}}$, also referred to as the ``splitting of the prior belief''. Formulation \eqref{eq:SplittingFormulation} provides an alternative point of view on the encoder's optimal utility \eqref{eq:PhiOptimalZ}.\\
$\bullet$ The optimal solution $\Phi_{\textsf{e}}^{\star} $ can be found by the concavification
method \cite{AM95}. In Sec \ref{sec:ExampleBinary}, we provide an example that illustrates the optimal splitting and the corresponding expected utility. \\
$\bullet$ When the channel is perfect and has a large input alphabet $|\mc{X}|\geq \min(|\mc{U}|,|\mc{V}|^{|\mc{Z}|})$, the strategic coding problem is equivalent to several i.i.d. copies of the one-shot problem, whose optimal solution is given by our characterization without the information constraint \eqref{eq:SplittingFormulation}. This noise-free setting is related to the problem of persuasion with heterogeneous beliefs under investigation in \cite{AlonsoCamara(JET)2016} and \cite{LaclauRenou17}.\\
$\bullet$ The information constraint $\sum_{w} \lambda_w \cdot h(p_w)  \geq H(U|Z)  - \max_{\PP(x)}I(X;Y)$ in \eqref{eq:SplittingFormulation} is a reformulation of $I(U;W|Z) \leq  \max_{\PP(x)}I(X;Y)$ in \eqref{eq:SetQ0}: 
\begin{align}
\sum_{w} \lambda_w \cdot h(p_w) =& \sum_{w} \lambda_w \cdot H(U|Z,W=w)\\
=& H(U|Z,W).
\end{align} 
$\bullet$ The dimension of the problem \eqref{eq:SplittingFormulation}  is $|\mc{U}|$. Caratheodory's Lemma (see \cite[Corollary 17.1.5, pp. 157]{rockafellar1970convex} and \cite[Corollary A.4, pp. 39]{LeTreustTomala17}) provides the cardinality bound: $|\mc{W}| = |\mc{U}|+1$. \\
$\bullet$ The cardinality of $\mc{W}$ is also restricted to the vector of recommended actions $|\mc{W}|  = |\mc{V}|^{|\mc{Z}|}$, telling to the decoder which action to play in each state. Otherwise assume that two posteriors $p_{w_1}$ and $p_{w_2}$  induce the same vectors of actions $v^1=(v_1^1,\ldots,v_{|\mc{Z}|}^1) = v^2=(v_1^2,\ldots,v_{|\mc{Z}|}^2)$. Then, both posteriors $p_{w_1}$ and $p_{w_2}$ can be replaced by their average: 
\begin{align}
\widetilde{p}  = \frac{\lambda_{w_1} \cdot p _{w_1} + \lambda_{w_2}\cdot p _{w_2}}{\lambda_{w_1} + \lambda_{w_2}},
\end{align}
without changing the utility and still satisfying the information constraint:
\begin{align}
& h(\widetilde{p})  \geq \frac{\lambda_{w_1} \cdot h(p _{w_1}) + \lambda_{w_2}\cdot h(p _{w_2})}{\lambda_{w_1} + \lambda_{w_2}}\label{eq:EntropyTilde}\\
\Longrightarrow& \sum_{w \neq w_1, \atop w \neq w_2} \lambda_{w} \cdot h(p_w) + (\lambda_{w_1} + \lambda_{w_2})\cdot h(\widetilde{p}) \geq H(U|Z) - \max_{\PP(x)}I(X;Y).
\end{align}
Inequality \eqref{eq:EntropyTilde} comes from the concavity of $h(p)$, stated in Lemma \ref{lemma:Concavity}.


Following the arguments of \cite[Theorem 5.1]{LeTreustTomala17}, the splitting under information constraint of Theorem \ref{theo:Concavification} can be reformulated in terms of Lagrangian and in terms of a general concavification $\tilde{\Psi}_{\textsf{e}}(p,\nu)$ defined by: 
\begin{align}
\tilde{\Psi}_{\textsf{e}}(p,\nu) &= 
\begin{cases}
\Psi_{\textsf{e}}(p), \text{ if } \nu \leq h(p),\\
- \infty, \text{ otherwise, }  
\end{cases}
\end{align}

\begin{theorem}\label{theo:Concavification2}
The optimal solution $\Phi_{\textsf{e}}^{\star}$ reformulates as:
\begin{align}
\Phi_{\textsf{e}}^{\star}  =& \inf_{t\geq 0 } \bigg\{ \cav \Big[\Psi_{\textsf{e}} + t \cdot h\Big]\big(\PP(u)\big) - t \cdot \Big(H(U|Z) - \max_{\PP(x)}I(X;Y)\Big)\bigg\}\label{eq:Lagrangian}\\
=&\cav \tilde{\Psi}_{\textsf{e}}\Big(\PP(u),H(U|Z) - \max_{\PP(x)}I(X;Y)\Big).\label{eq:GeneralCav}
\end{align}
\end{theorem}
Equation \eqref{eq:Lagrangian} is the concavification of a Lagrangian that integrates the information constraint. The proof follows directly from \cite[Proposition A.2, pp. 37]{LeTreustTomala17}.\\
Equation \eqref{eq:GeneralCav} corresponds to a bi-variate concavification where the information constraint requires an additional dimension. The proof follows directly from \cite[Lemma A.1, pp. 36]{LeTreustTomala17}.

%


\section{Example with Binary Source and State}\label{sec:ExampleBinary}

We consider a binary source $U \in \{u_1,u_2\}$ with probability $\PP(u_2)=p_0 \in[0,1]$. The binary state information $Z\in \{z_1,z_2\}$ is drawn according to the conditional probability distribution $\PP(z|u)$ with parameter $\delta_1\in [0,1]$ and $\delta_2\in [0,1]$, as depicted in Fig. \ref{fig:SignalingZ}. For the clarity of the presentation, we consider a binary auxiliary random variable $W \in \{w_1,w_2\}$, even if this choice might be sub-optimal since: $ |\mc{W}| =\min\big(|\mc{U}|+1, |\mc{V}|^{|\mc{Z}|}\big) =3$. The solutions provided  below implicitly refer to this special case of $ |\mc{W}|=2$. The random variable $W$ is drawn according to the conditional probability distribution $\QQ(w|u)$ with parameters $\alpha \in [0,1]$ and $\beta\in [0,1]$. The joint probability distribution $\PP(u,z)\times \QQ(w|u)$ is represented by Fig. \ref{fig:SignalingZ}.
\begin{figure}[!ht]
\begin{center}
\psset{xunit=0.9cm,yunit=0.9cm}
\begin{pspicture}(-2.5,0)(3,3)
\rput(0,0.5){$u_2$}
\rput(0,2.5){$u_1$}
\psdots(0.5, 0.5)(0.5, 2.5)(-0.5, 0.5)(-0.5, 2.5)(3.5,0.5)(3.5,2.5)(-3.5,0.5)(-3.5,2.5)
\psline[linewidth=1pt]{<-}(-3.5,2.5)(-0.5,2.5)
\psline[linewidth=1pt]{<-}(-3.5,0.5)(-0.5,0.5)
\psline[linewidth=1pt]{<-}(-3.5,0.5)(-0.5,2.5)
\psline[linewidth=1pt]{<-}(-3.5,2.5)(-0.5,0.5)
\psline[linewidth=1pt]{->}(0.5,0.5)(3.5,0.5)
\psline[linewidth=1pt]{->}(0.5,2.5)(3.5,2.5)
\psline[linewidth=1pt]{->}(0.5,0.5)(3.5,2.5)
\psline[linewidth=1pt]{->}(0.5,2.5)(3.5,0.5)
\rput(4,0.5){$w_2$}
\rput(4,2.5){$w_1$}
\rput(-4,0.5){$z_2$}
\rput(-4,2.5){$z_1$}
\rput(2,3){$1 - \alpha$}
\rput(2,0){$1 - \beta$}
\rput(0,3){$1-p_0$}
\rput(0,0){$p_0$}
\rput(0.8,1.9){$ \alpha$}
\rput(0.8,1.1){$ \beta$}
\rput(-2,3){$1 - \delta_1$}
\rput(-0.8,1.9){$ \delta_1$}
\rput(-2,0){$1 -  \delta_2$}
\rput(-0.8,1.1){$ \delta_2$}
\end{pspicture}
\caption{Joint probability distribution $\PP(u,z)\times \QQ(w|u)$ depending on parameters $p_0\in[0,1]$, $\delta_1\in [0,1]$, $\delta_2\in [0,1]$, $\alpha\in [0,1]$ and $\beta\in [0,1]$.}
\label{fig:SignalingZ}
\end{center}
\end{figure}
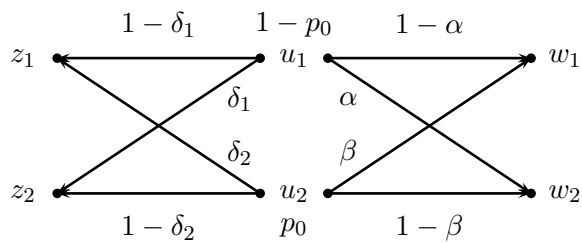
The utility functions of the encoder $\phi_{\textsf{e}}(u,v)$ and decoder $\phi_{\textsf{d}}(u,v)$ are given by Fig. \ref{fig:Utility1}, \ref{fig:Utility2} and do not depend on the state $z$. In this example, the goal of the decoder is to choose the action $v$ that matches the source symbol $u$, whereas the goal of encoder is to persuade the decoder to take the action $v_2$. 
\begin{figure}[ht!]
\begin{minipage}{0.495\linewidth}
\begin{center}
\psset{xunit=1cm,yunit=1cm}
\begin{pspicture}(0,0.2)(2,2.2)
\psframe(0,0)(2,2)
\psline(1,0)(1,2)
\psline(0,1)(2,1)
\rput(-0.3,0.5){$u_2$}
\rput(-0.3,1.5){$u_1$}
\rput(0.5,2.3){$\textcolor[rgb]{0,0.6,0}{v_1}$}
\rput(1.5,2.3){$\textcolor[rgb]{0,0,1}{v_2}$}
\rput(0.5,1.5){$\textcolor[rgb]{0,0.6,0}{0}$}
\rput(0.5,0.5){$\textcolor[rgb]{0,0.6,0}{0}$}
\rput(1.5,1.5){$\textcolor[rgb]{0,0,1}{1}$}
\rput(1.5,0.5){$\textcolor[rgb]{0,0,1}{1}$}
\end{pspicture}
\caption{Utility function of the encoder $\phi_{\textsf{e}}(u,v)$.}\label{fig:Utility1}
\end{center}
\end{minipage}\hfill
\begin{minipage}{0.495\linewidth}
\begin{center}
\psset{xunit=1cm,yunit=1cm}
\begin{pspicture}(0,0.2)(2,2.2)
\psframe(0,0)(2,2)
\psline(1,0)(1,2)
\psline(0,1)(2,1)
\rput(-0.3,0.5){$u_2$}
\rput(-0.3,1.5){$u_1$}
\rput(0.5,2.3){$\textcolor[rgb]{0,0.6,0}{v_1}$}
\rput(1.5,2.3){$\textcolor[rgb]{0,0,1}{v_2}$}
\rput(0.5,1.5){$\textcolor[rgb]{0,0.6,0}{9}$}
\rput(0.5,0.5){$\textcolor[rgb]{0,0.6,0}{4}$}
\rput(1.5,1.5){$\textcolor[rgb]{0,0,1}{0}$}
\rput(1.5,0.5){$\textcolor[rgb]{0,0,1}{10}$}
\end{pspicture}
\caption{Utility function of the decoder $\phi_{\textsf{d}}(u,v)$.}\label{fig:Utility2}
\end{center}
\end{minipage}\hfill
\end{figure}
After receiving the pair of symbols $(w,z)$, the decoder updates its \emph{posterior belief} $\PP(\cdot|w,z)\in \Delta(\mc{U})$, according to Bayes rule. We denote by $p\in \Delta(\mc{U})$ the decoder's belief and we denote by $\gamma=0.6$ the \emph{belief threshold} at which the decoder changes its action. When the decoder's belief is exactly equal to the threshold $p(u_2) = \gamma = 0.6$, the decoder is indifferent between the two actions $\{v_1,v_2\}$ so it chooses $v_1$, i.e. the worst action for the encoder.  Hence the decoder chooses a best-reply action $v_1^{\star}$ or $v_2^{\star}$ depending on the interval $[0,0.6]$ or $]0.6,1]$ in which lies the belief $p(u_2)\in[0,1]$, see Fig. \ref{fig:BestReply}.
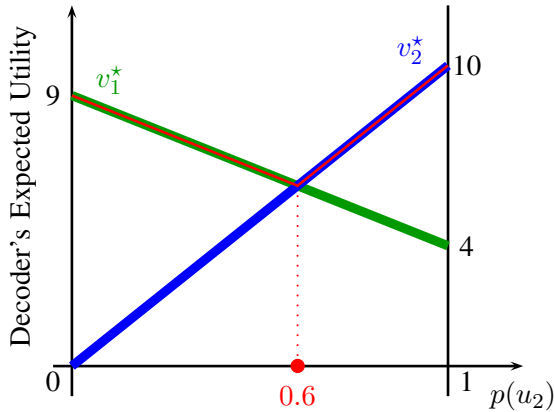
\begin{figure}[!ht]
\begin{center}
\psset{xunit=0.5cm,yunit=0.4cm}
\begin{pspicture}(-1,-1)(10,12)
\psline[linewidth=1pt]{->}(0,-1)(0,12)
\psline[linewidth=1pt]{->}(-0.5,0)(12,0)
\psline[linewidth=1pt](10,-1)(10,12)
\rput(-0.5,-0.5){0}
\rput(10.5,-0.5){1}
\rput(12,-1){$p(u_2)$}
\rput{90}(-1.3,6){\normalsize\textcolor[rgb]{0.00,0.00,0.00}{Decoder's Expected Utility}}
\rput[u](1,9.5){$\textcolor[rgb]{0,0.6,0}{v_1^{\star}}$}
\psline[linecolor=vert, linewidth=3.5pt](0,9)(10,4)
\rput(-0.5,9){9}
\rput(10.5,4){4}
\rput(9,10.5){$\textcolor[rgb]{0,0,1}{v_2^{\star}}$}
\psline[linecolor=blue, linewidth=3.5pt](0,0)(10,10)
\rput(10.5,10){10}
\psdots[linecolor = red,dotscale=1.5](6 ,0)
\rput(6,-1){$\textcolor[rgb]{1,0,0}{0.6}$}
\psline[linecolor=red, linewidth=1pt](0,9)(6,6)(10,10)
\psline[linestyle=dotted,linecolor=red](6,0)(6,6)
\end{pspicture}
\caption{The decoder's best-reply action $v^{\star}$ depends on its belief $p\in\Delta(\mc{U})$: it plays $\textcolor[rgb]{0,0.6,0}{v_1^{\star}}$ if $p(u_2)\in[0,0.6]$ and $\textcolor[rgb]{0,0,1}{v_2^{\star}}$ if $p(u_2)\in]0.6,1]$.}\label{fig:BestReply}
\end{center}
\end{figure}
Since the utility functions given by Fig. \ref{fig:Utility1}, \ref{fig:Utility2} do not depends on $Z$, we denote by $\psi_{\textsf{e}}(p)$ the robust utility function of the encoder (see Definition \ref{def:RobustUtility}), given by:
\begin{align}
\psi_{\textsf{e}}(p) &=  \min_{v \in \argmax \atop p(u_1)\cdot\phi_{\textsf{d}}(u_1,v) +  p(u_2) \cdot\phi_{\textsf{d}}(u_2,v)} p(u_1) \cdot\phi_{\textsf{e}}(u_1,v) + p(u_2) \cdot\phi_{\textsf{e}}(u_2,v),\\
&= \UN\Big(p(u_2) \in]0.6,1]\Big).\label{eq:FunctionPsi1}
\end{align}


\subsection{Concavification without Information Constraint}\label{sec:ExampleConcavificationNo}

In this section, we assume that the channel capacity is large enough $ \max_{\PP(x)}I(X;Y)\geq \log_2|\mc{U}|$, so we investigate the concavification of $ \Psi_{\textsf{e}}(p)$ (see Definition \ref{def:AverageUtility}), without information constraint:
\begin{align}
\Phi_{\textsf{e}}^{\circ}  =&  \sup\bigg\{ \sum_{w} \lambda_w \cdot \Psi_{\textsf{e}}(p_w)   \quad \text{ s.t. } \quad \sum_{w} \lambda_w \cdot p_w  = \PP(u) \in \Delta(\mc{U}) \bigg\}.\label{eq:SplittingFormulationC}
\end{align}
The correlation between random variables $(U,Z)$ is fixed whereas the correlation between random variables  $(U,W)$ is chosen strategically by the encoder. This imposes a strong relationship between the three different kinds of posterior beliefs: $\QQ(u|z)$, $\QQ(u|w)$, $\QQ(u|w,z)$. We denote by $p_{01}\in [0,1]$, $p_{02}\in [0,1]$ the belief parameters after observing the state information $Z$ only:
\begin{align}
p_{01} &= \QQ(u_2 | z_1) = \frac{p_0 \cdot \delta_2}{p_0 \cdot \delta_2 + (1-p_0) \cdot (1- \delta_1)},\label{eq:ex-ante1}\\
p_{02} &= \QQ(u_2 | z_2) = \frac{p_0 \cdot (1-\delta_2)}{p_0 \cdot (1-\delta_2)  + (1-p_0) \cdot \delta_1}.\label{eq:ex-ante2}
\end{align}
The beliefs parameters $p_{01}$, $p_{02}$ are given by the source probability distribution $\PP(u,z)\in\Delta(\mc{U} \times \mc{Z})$ and correspond to the horizontal dotted lines in Fig. \ref{fig:PlayerOneUtility0.5_0.7_0.9_0.6}, for $p_0 = 0.5$, $\delta_1 = 0.7$, $\delta_2 = 0.9$. We denote by $q_1\in [0,1]$, $q_2\in [0,1]$ the belief parameters after observing only the symbol $W$ sent by the encoder:
\begin{align}
q_1 &= \QQ(u_2 | w_1) = \frac{p_0 \cdot \beta}{p_0 \cdot \beta + (1-p_0) \cdot (1- \alpha)},\label{eq:ex-ante3}\\
q_2 &= \QQ(u_2 | w_2) = \frac{p_0 \cdot (1-\beta)}{p_0 \cdot (1-\beta)  + (1-p_0) \cdot \alpha}.\label{eq:ex-ante4}
\end{align}
By inverting the system of equations \eqref{eq:ex-ante3} - \eqref{eq:ex-ante4}, we express the cross-over probabilities $\alpha(q_1,q_2)$ and $\beta(q_1,q_2)$ as functions of the target belief parameters $(q_1,q_2)$:
\begin{align}
\begin{cases}
\alpha(q_1,q_2) &= \frac{(1 - q_2) \cdot (q_1   - p_0 )}{(1-p_0) \cdot (q_1 - q_2) }\\
\beta(q_1,q_2) &= \frac{q_1 \cdot ( p_0  - q_2 ) }{   p_0 \cdot(q_1 - q_2) }\\
\end{cases}
\end{align}
\begin{lemma}[Feasible Posteriors]\label{lemma:FeasiblePosteriors}
The parameters $\alpha(q_1,q_2)$ and $\beta(q_1,q_2)$ belong to the interval $[0,1]$ if and only if $q_1 \leq p_0 \leq q_2$ or $q_2 \leq p_0 \leq q_1$.
\end{lemma}
The proof of Lemma \ref{lemma:FeasiblePosteriors} is stated in App. \ref{sec:ProofLemmaPosteriors}. Thanks to the Markov chain property $Z  -\!\!\!\!\minuso\!\!\!\!-U    -\!\!\!\!\minuso\!\!\!\!-  W$ the posterior beliefs $\QQ(u|w,z)$ reformulate in terms of $\QQ(u|w)$:
\begin{align}
\QQ(u|w,z)  =&  \frac{\QQ(u,z,w)  }{\QQ(z,w)  } = \frac{\QQ(u,z|w)  }{\sum_{u'}\QQ(u',z|w)  }= \frac{\QQ(u|w) \PP(z|u)  }{\sum_{u'}\QQ(u'|w) \PP(z|u') }, \qquad \forall (u,z,w)\in \mc{U}\times \mc{Z}\times \mc{W} .\label{eq:ReformulationPosteriors}
\end{align}
\FigDynamic{0.5}{0.7}{0.9}{0.6}
We define the belief parameters $p_1\in[0,1]$, $p_2\in[0,1]$, $p_3\in[0,1]$, $p_4\in[0,1]$ after observing $(W,Z)$ and we express them as functions of $q_1$, $q_2$:
\begin{align}
p_1   &= \QQ(u_2 | w_1,z_1) = \frac{q_1 \cdot \delta_2}{ (1-q_1)\cdot(1-\delta_1)  + q_1 \cdot \delta_2  },\\
p_2   &= \QQ(u_2 | w_1,z_2) = \frac{q_1 \cdot (1-\delta_2)}{ (1-q_1)\cdot \delta_1  +q_1 \cdot (1-\delta_2) },\\
p_3   &= \QQ(u_2 | w_2,z_1) = \frac{q_2 \cdot \delta_2 }{(1-q_2) \cdot (1-\delta_1)  + q_2 \cdot \delta_2 },\\
p_4   &= \QQ(u_2 | w_2,z_2) = \frac{q_2 \cdot (1-\delta_2)}{ (1-q_2) \cdot \delta_1 + q_2 \cdot (1-\delta_2) }.
\end{align}
Fig. \ref{fig:Dynamic0.5_0.7_0.9_0.6} represents $(p_1,p_2,p_3,p_4)$ as a functions of $q_1\in[0,p_0]$ and $q_2\in[p_0,1]$. In fact, the beliefs $q_1$ and $q_2$, are the \emph{key parameters} since they control the decoder's best-reply action $v^{\star}(p)$ through the beliefs $(p_1,p_2,p_3,p_4)$. We define the two following functions of $q\in[0,1]$:
\begin{align}
p_1(q)   &=  \frac{q \cdot \delta_2}{ (1-q)\cdot(1-\delta_1)  + q \cdot \delta_2  },\\
p_2(q)   &=  \frac{q \cdot (1-\delta_2)}{ (1-q)\cdot \delta_1  +q \cdot (1-\delta_2) }.
\end{align}
Given the \emph{belief threshold} $\gamma=0.6$ at which the decoder changes its action, we define the parameters $\nu_1$ and $\nu_2$ such that $p_1(\nu_1) = \gamma$ and $p_2(\nu_2) = \gamma$.
\begin{align}
\gamma = p_1(\nu_1) \;\; \Longleftrightarrow&\;\;  \nu_1 = \frac{\gamma \cdot (1 - \delta_1)}{\delta_2 \cdot (1 - \gamma) + \gamma \cdot (1 - \delta_1)},\\
\gamma = p_2(\nu_2) \;\; \Longleftrightarrow&\;\;  \nu_2 = \frac{\gamma \cdot   \delta_1}{\gamma \cdot   \delta_1  + (1 - \delta_2) \cdot (1 - \gamma) }.
\end{align}
This belief parameters $\nu_1$ and $\nu_2$ allow to reformulate the utility function of the encoder as a function of belief $\QQ(u|w)$ (see Fig. \ref{fig:PlayerOneAuxUtility0.5_0.7_0.9_0.6}), instead of belief $\QQ(u|w,z)$ (see Fig. \ref{fig:PlayerOneUtility0.5_0.7_0.9_0.6}). 
\FigPlayerOneAuxUtility{0.5}{0.7}{0.9}{0.6}

The solution $\Phi_{\textsf{e}}^{\circ}$ corresponds to the concavification of the function $\Psi_{\textsf{e}}$, defined over $q\in[0,1]$:
\begin{align}
\Psi_{\textsf{e}}(q) &= \Big( (1-q)\cdot(1-\delta_1)  + q \cdot \delta_2\Big) \cdot \psi_{\textsf{e}}\big(p_1(q)\big) + \Big( (1-q)\cdot \delta_1  +q \cdot (1-\delta_2) \Big)\cdot\psi_{\textsf{e}}\big(p_2(q)\big),\\
\Phi_{\textsf{e}}^{\circ} &= \cav \Psi_{\textsf{e}}(p_0) = \sup_{\lambda,\atop q,q'}\bigg\{\lambda \cdot \Psi_{\textsf{e}}(q) + (1-\lambda) \cdot \Psi_{\textsf{e}}(q') \quad \text{ s.t. } \quad \lambda \cdot q + (1-\lambda) \cdot q' = p_0 \bigg\}.
\end{align}
\FigPlayerOneUtility{0.5}{0.7}{0.9}{0.6}
Fig. \ref{fig:PlayerOneAuxUtility0.5_0.7_0.9_0.6} represents the utility function $\Psi_{\textsf{e}}(q)$ of the encoder, depending on the belief $q\in[0,1]$. When the belief $q\in [\nu_1,\nu_2]$, then $\psi_{\textsf{e}}\big(p_1(q)\big) = 1$ whereas $\psi_{\textsf{e}}\big(p_2(q)\big)=0$. The optimal splitting is represented  by the square and circle. This indicates that the optimal  posterior beliefs are $(p_1,p_2) = (\gamma, p_2(\nu_1))$ and  $(p_3,p_4) = ( p_1(\nu_2),\gamma)$, as in Fig. \ref{fig:PlayerOneUtility0.5_0.7_0.9_0.6}. 

When the decoder has no state information, the optimal solution  by Kamenica-Gentzkow \cite{KamenicaGentzkow11} is the concavification of the function $\psi_{\textsf{e}}(p)  = \UN\big(p\in]0.6,1]\big)$, corresponding to $\widetilde{\Phi_{\textsf{e}}}$ in Fig. \ref{fig:PlayerOneUtility0.5_0.7_0.9_0.6}. In this example, the decoder's state information $Z$ decreases the optimal utility of the encoder $\widetilde{\Phi_{\textsf{e}}} \geq \Phi_{\textsf{e}}^{\circ}$.


\subsection{Concavification with Information constraint}\label{sec:concavificationIC}

In this section, we assume that the channel capacity is equal to: $C= \max_{\PP(x)}I(X;Y)=0.1$, so the information constraint of $\Q_0$ is active. The average entropy stated in \eqref{eq:FunctionHp} reformulates as a function of $q\in[0,1]$:
\begin{align}
h(q) =& \Big( (1-q)\cdot(1-\delta_1)  + q \cdot \delta_2\Big) \cdot H_b\big(p_1(q)\big) + \Big( (1-q)\cdot \delta_1  +q \cdot (1-\delta_2) \Big)\cdot H_b\big(p_2(q)\big),\label{eq:AverageEntropy}
\end{align} 
where $H_b(\cdot)$ denotes the binary entropy. The dark blue region in Fig. \ref{fig:PosteriorsRegion0.5_0.7_0.9_0.6_0.1} represents the set of posterior distributions $(q_1,q_2)$ with $q_1 \leq p_0 \leq q_2$ or $q_2 \leq p_0 \leq q_1$, that satisfy the information constraint:
\begin{align}
&\frac{p_0 - q_2}{q_1 - q_2} \cdot h(q_1) + \frac{q_1 - p_0}{q_1 - q_2} \cdot h(q_2) \geq H(U|Z) -  \max_{\PP(x)}I(X;Y)\\
\Longleftrightarrow & \QQ(w_1) \cdot H(U|Z,W=w_1) + \QQ(w_2) \cdot H(U|Z,W=w_2) \geq H(U|Z) -  \max_{\PP(x)}I(X;Y)\\
\Longleftrightarrow &I(U;W|Z) \leq \max_{\PP(x)}I(X;Y).
\end{align} 
\FigPosteriorsRegion{0.5}{0.7}{0.9}{0.6}{0.1}
The green region in Fig. \ref{fig:PosteriorsRegion0.5_0.7_0.9_0.6_0.1} corresponds to the information constraint $I(U;W) \leq \max_{\PP(x)}I(X;Y)$, i.e. when the decoder does not observe the state information $Z$. The case without decoder's state information $Z$ is investigated in \cite{LeTreustTomala17} and the corresponding information constraint is given by:
\begin{align}
&\frac{p_0 - q_2}{q_1 - q_2} \cdot H_b(q_1) + \frac{q_1 - p_0}{q_1 - q_2} \cdot H_b(q_2) \geq H(U) -  \max_{\PP(x)}I(X;Y)  \\
\Longleftrightarrow &  \QQ(w_1) \cdot H(U|W=w_1) + \QQ(w_2) \cdot H(U|W=w_2)  \geq H(U) -  \max_{\PP(x)}I(X;Y) \\
\Longleftrightarrow & I(U;W) \leq \max_{\PP(x)}I(X;Y).
\end{align} 
Fig. \ref{fig:PosteriorsRegion0.5_0.7_0.9_0.6_0.1} shows that the decoder's state information $Z$ enlarges the set of posterior beliefs $\QQ(u|w)$ compatible with the information constraint of $\Q_0$. 
\FigPlayerOneAuxUtilityEntropyBis{0.5}{0.7}{0.9}{0.6}{0.1}

The optimal posterior beliefs are represented on Fig. \ref{fig:PlayerOneAuxUtilityEntropyBis0.5_0.7_0.9_0.6_0.1}, by the square and the circle.  Due to the restriction imposed by the channel, the optimal posterior beliefs are $(\nu_3, \nu_2)$  instead of $(\nu_1, \nu_2)$, and this reduces the encoder's optimal utility to $\Phi_{\textsf{e}}^{\star} \simeq 0.63$ instead of $\Phi_{\textsf{e}}^{\circ} \simeq 0.64$. In fact, the posterior beliefs $(\nu_1, \nu_2)$ do not satisfy the information constraint: $\lambda h(\nu_1) + (1-\lambda) h(\nu_2) < H(U|Z) -  \max_{\PP(x)}I(X;Y)$, whereas the pair of posterior beliefs $(\nu_3, \nu_2)$ lies at the boundary of the blue region in Fig. \ref{fig:PosteriorsRegion0.5_0.7_0.9_0.6_0.1}.  
\FigPlayerOneUtilityEntropy{0.5}{0.7}{0.9}{0.6}{0.1}
Posterior beliefs $(\nu_3, \nu_2)$ determine the posterior beliefs $(p_1,p_2,p_3,p_4)$, represented in Fig. \ref{fig:PlayerOneUtilityEntropy0.5_0.7_0.9_0.6_0.1}, that satisfy the reformulation of the information constraint:
\begin{align}
\lambda_{w_1,z_1} H_b(p_1) + \lambda_{w_1,z_2} H_b(p_2) + \lambda_{w_2,z_1} H_b(p_3) + \lambda_{w_2,z_2} H_b(p_4) = H(U|Z) - \max_{\PP(x)}I(X;Y),
\end{align}
and provide the corresponding expected utility $\Phi_{\textsf{e}}^{\star} \simeq 0.63$.

\begin{remark}[Impact of the State Information $Z$]
The state information $Z$ at the decoder has two effects:\\
$\bullet$ When the communication is restricted (i.e. $\max_{\PP(x)}I(X;Y)<\log_2|\mc{U}|$), it enlarges the set of posterior beliefs $\QQ(u|w)$, so it may increase the encoder's utility.\\
$\bullet$ Since it reveals partial information to the decoder, it forces the decoder to choose a best-reply actions that might be sub-optimal for the encoder, so it may decrease the encoder's utility.\\
Depending on the problem, the state information $Z$ may increase or decrease the encoder's optimal utility.
\end{remark}

%
%
%
%


\appendices



\section{Proof of Theorem \ref{theo:Concavification}}\label{sec:ProofConcavification}

We identify the weight parameters $\lambda_w =\QQ(w)$ and $p_w = \QQ(u|w)\in \Delta(\mc{U})$ and   \eqref{eq:SplittingFormulation} becomes:
\begin{align}
&\sup_{\lambda_w\in [0,1],\atop p_w\in \Delta(\mc{U})}\bigg\{ \sum_{w} \lambda_w \cdot \Psi_{\textsf{e}}(p_w)   \quad \text{ s.t. } \quad \sum_{w} \lambda_w \cdot p_w  = \PP(u) \in \Delta(\mc{U}), \nonumber \\
&\qquad\qquad\qquad\qquad\qquad \text{ and }\quad \sum_{w} \lambda_w \cdot h(p_w)  \geq H(U|Z)  - \max_{\PP(x)}I(X;Y) \bigg\}\label{eq:SplittingFormulation1}\\
=&\sup_{ \QQ(w) , \QQ(u|w)}\bigg\{ \sum_{w} \QQ(w) \cdot \Psi_{\textsf{e}}(\QQ(u|w))   \quad \text{ s.t. } \quad \sum_{w} \QQ(w) \cdot \QQ(u|w)  = \PP(u) \in \Delta(\mc{U}),\nonumber\\
&\qquad\qquad\qquad\qquad\qquad \text{ and }\quad \sum_{w} \QQ(w)  \cdot h(\QQ(u|w) )  \geq H(U|Z)  - \max_{\PP(x)}I(X;Y) \bigg\}\label{eq:SplittingFormulation2}\\
=&\sup_{ \QQ(w) , \QQ(u|w)}\bigg\{ \sum_{w} \QQ(w) \cdot  \sum_{u,z}  \QQ(u|w)  \cdot \PP(z|u)\cdot  \psi_{\textsf{e}}\bigg(z,\QQ(u|w,z) \bigg),\nonumber\\
&  \qquad \text{ s.t. } \quad \sum_{w} \QQ(w) \cdot \QQ(u|w)  = \PP(u) \in \Delta(\mc{U}),\nonumber\\
&\qquad  \text{ and }\quad \sum_{w} \QQ(w)  \cdot  \sum_{u,z} \QQ(u|w) \cdot \PP(z|u)\cdot \log_2 \frac{1}{\QQ(u|w,z) }  \geq H(U|Z)  - \max_{\PP(x)}I(X;Y) \bigg\}\label{eq:SplittingFormulation4}\\
=&\sup_{ \QQ(w) , \QQ(u|w)}\bigg\{ \sum_{w} \QQ(w) \cdot  \sum_{u,z}  \QQ(u|w)  \cdot \PP(z|u)\cdot  \psi_{\textsf{e}}\bigg(z,\QQ(u|w,z) \bigg),\nonumber\\
&  \qquad \text{ s.t. }  \sum_{w} \QQ(w) \cdot \QQ(u|w)  = \PP(u) \in \Delta(\mc{U}),   \text{ and } H(U|W,Z)  \geq H(U|Z)  - \max_{\PP(x)}I(X;Y) \bigg\}\label{eq:SplittingFormulation5}\\
=&\sup_{ \QQ(u,z,w) \in \Q_0}\E_{\QQ(u,z,w) } \bigg[\phi_{\textsf{e}}(U,Z,V^{\star}(z,\QQ(u|w,z)))\bigg]\label{eq:SplittingFormulation6}\\
=&\sup_{ \QQ(u,z,w) \in \Q_0} \min_{\QQ(v |z,w) \in  \atop \Q_2(\QQ(u,z,w))} \E_{\QQ(u,z,w) \atop \times \QQ(v |z,w)} \bigg[\phi_{\textsf{e}}(U,Z,V)\bigg] = \Phi_{\textsf{e}}^{\star}.\label{eq:SplittingFormulation7}
\end{align}
Equation \eqref{eq:SplittingFormulation2} comes from the identification of the weight parameters $\lambda_w =\QQ(w)$ and $p_w = \QQ(u|w)\in \Delta(\mc{U})$.\\
Equation \eqref{eq:SplittingFormulation4} comes from the definitions of $\Psi_{\textsf{e}}(p_w) $ and $h(p_w) $ in \eqref{eq:FunctionPsip} and \eqref{eq:FunctionHp} and from: $\QQ(u|w,z) = \frac{\QQ(u,z|w)}{\QQ(z|w)} = \frac{\QQ(u|w) \cdot  \PP(z|u)}{\sum_{u'}\QQ(u'|w) \cdot  \PP(z|u')}$, due to the  Markov chain property $Z  -\!\!\!\!\minuso\!\!\!\!-U    -\!\!\!\!\minuso\!\!\!\!-  W$ and the fixed distribution of the source $\PP(u,z) $.\\
Equations \eqref{eq:SplittingFormulation5} - \eqref{eq:SplittingFormulation7} are reformulations.


\section{Achievability Proof of Theorem \ref{theo:MaxMinStackelberg}}\label{sec:AchievabilityProof}

This proof is built on Wyner-Ziv's source coding \cite{wyner-it-1976} and the achievability proof stated in \cite[App. A.3.2, pp. 44]{LeTreustTomala17} with one additional feature: we show that the average posterior beliefs induced by Wyner-Ziv's source coding converge to the target posterior beliefs.


\subsection{Zero Capacity}\label{sec:ZeroCapacity}

We first investigate the case of zero capacity.
\begin{lemma}\label{lemma:ZeroCapacity}
If the channel has zero capacity: $\max_{\PP(x)} I( X; Y ) =0$, then  we have:
\begin{align}
\forall n \in \N,\; \forall \sigma,\qquad  \min_{\tau \in \textsf{BR}_{\textsf{d}}(\sigma)} \Phi_{\textsf{e}}^n(\sigma, \tau)  =  \Phi_{\textsf{e}}^{\star}.
\end{align}
\end{lemma}

\begin{proof}[Lemma \ref{lemma:ZeroCapacity}]
The zero capacity $\max_{\PP(x)} I( X; Y ) =0$ implies that any probability distribution $ \PP(u,z) \times \QQ(w | u) \in \Q_0$ satisfies $I( U ;W |Z ) =0$, corresponding to the Markov chain property $U -\!\!\!\!\minuso\!\!\!\!- Z -\!\!\!\!\minuso\!\!\!\!- W$, i.e. $\QQ(u|z,w)=\PP(u|z)$ for all $(u,z,w)\in \mc{U}\times  \mc{Z}\times  \mc{W} $. 
\begin{align}
 \Phi_{\textsf{e}}^{\star}=&\sup_{ \QQ(u,z,w) \in \Q_0} \min_{\QQ(v |z,w) \in  \atop \Q_2(\QQ(u,z,w))} \E_{\QQ(u,z,w) \atop \times \QQ(v |z,w)} \bigg[\phi_{\textsf{e}}(U,Z,V)\bigg] \label{eq:LemmaZero1}\\
=& \sup_{ \QQ(u,z,w) \in \Q_0}\E_{\QQ(u,z,w) } \bigg[\phi_{\textsf{e}}(U,Z,V^{\star}(z,\QQ(u|w,z)))\bigg]\label{eq:LemmaZero2}\\
=& \sup_{ \QQ(u,z,w) \in \Q_0}\E_{\QQ(u,z,w) } \bigg[\phi_{\textsf{e}}(U,Z,V^{\star}(z,\PP(u|z)))\bigg]\label{eq:LemmaZero3}\\
=& \E_{\PP(u,z) } \bigg[\phi_{\textsf{e}}(U,Z,V^{\star}(z,\PP(u|z)))\bigg]\label{eq:LemmaZero4}.
\end{align}
Equation \eqref{eq:LemmaZero2} is a reformulation by using the best-reply action $v^{\star}\big(z,p\big)$ of Definition \ref{def:BestReply} for symbol $z\in \mc{Z}$ and the belief $\QQ(u|w,z)$.\\
Equation \eqref{eq:LemmaZero3} comes from Markov chain property $U -\!\!\!\!\minuso\!\!\!\!- Z -\!\!\!\!\minuso\!\!\!\!- W$ that allows to replace the belief $\QQ(u|w,z)$ by $\PP(u|z)$.\\
Equation \eqref{eq:LemmaZero4} comes from removing the random variable $W$ since it has no impact on the function $\phi_{\textsf{e}}(u,z,v^{\star}(z,\PP(u|z)))$.\\ 

For any $n$ and for any encoding strategy $\sigma$, the encoder's long-run utility is given by:
\begin{align}
\min_{\tau \in \textsf{BR}_{\textsf{d}}(\sigma)}\Phi_{\textsf{e}}^n(\sigma, \tau) 
&=\min_{\tau \in \textsf{BR}_{\textsf{d}}(\sigma)}\sum_{u^n,z^n,x^n,\atop y^n,v^n}\prod_{i=1}^n\PP\big(u_i,z_i \big) \times\sigma\big(x^n\big| u^n \big)\times \prod_{i=1}^n \mc{T}\big(y_i\big) \times \tau\big(v^n \big| y^n ,z^n\big) \cdot  \Bigg[  \frac{1}{n} \sum_{i=1}^n \phi_{\textsf{e}}(u_i,z_i,v_i) \Bigg]\label{eq:LemmaZero5}\\
&=\min_{\tau \in \textsf{BR}_{\textsf{d}}(\sigma)}\sum_{u^n,z^n,v^n}\prod_{i=1}^n\PP\big(u_i,z_i \big) \times  \tau\big(v^n \big|z^n\big) \cdot  \Bigg[  \frac{1}{n} \sum_{i=1}^n \phi_{\textsf{e}}(u_i,z_i,v_i) \Bigg]\label{eq:LemmaZero6}\\
&= \frac{1}{n} \sum_{i=1}^n \Bigg[ \sum_{u_i,z_i,v_i} \PP\big(u_i,z_i \big) \times \UN(v_i^{\star}(z_i,\QQ(u|z))) \cdot   \phi_{\textsf{e}}(u_i,z_i,v_i)\Bigg] \label{eq:LemmaZero7}\\
&=  \E_{\PP(u,z) } \bigg[\phi_{\textsf{e}}(U,Z,V^{\star}(z,\PP(u|z)))\bigg].\label{eq:LemmaZero9}
\end{align}
Equation \eqref{eq:LemmaZero5} comes from the zero capacity that imposes that the channel outputs $Y^n$ are independent of the channel inputs $X^n$.\\
Equation \eqref{eq:LemmaZero6} comes from removing the random variables $(X^n,Y^n)$ and noting that the decoder's best-reply $ \tau\big(v^n \big|z^n\big)$ does not depend on $y^n$ anymore, since $y^n$ is independent of $(u^n,z^n)$.\\
Equation \eqref{eq:LemmaZero7} is a reformulation based on the best-reply action $v^{\star}\big(z,\PP(u|z)\big)$ of Definition \ref{def:BestReply}, for the symbol $z\in \mc{Z}$ and the belief $\PP(u|z)$.\\
Equation \eqref{eq:LemmaZero9} comes from the i.i.d. property of $(U,Z)$ and  concludes the proof of Lemma \ref{lemma:ZeroCapacity}.
\end{proof}


\subsection{Positive Capacity}\label{sec:PositiveCapacity}

We now assume that the channel capacity is strictly positive $\max_{\PP(x)} I( X; Y ) >0$. We consider an auxiliary concavification in which the information constraint is satisfied with \emph{strict} inequality and the sets of decoder's best-reply actions are always \emph{singletons}.
\begin{align}
\widehat{\Phi_{\textsf{e}}}  =&  \sup\bigg\{ \sum_{w} \lambda_w \cdot \Psi_{\textsf{e}}(p_w)   \quad \text{ s.t. } \quad \sum_{w} \lambda_w \cdot p_w  = \PP(u) \in \Delta(\mc{U}),\nonumber\\
&\qquad\qquad\qquad\qquad\qquad \text{ and }\quad \sum_{w} \lambda_w \cdot h(p_w)  > H(U|Z)  - \max_{\PP(x)}I(X;Y),\nonumber\\
&\qquad\qquad\qquad\qquad\qquad \text{ and }\quad \forall (z,w)\in \mc{Z}\times\mc{W}, \;\; \mc{V}^{\star}(z,\QQ(u|z,w)) \;\; \text{is a singleton} \;\; \bigg\},\label{eq:SplittingFormulationHat}
\end{align}

\begin{lemma}\label{lemma:RestrictedSplitting}
If $\max_{\PP(x)} I( X; Y ) >0$, then $\widehat{\Phi_{\textsf{e}}} = \Phi_{\textsf{e}}^{\star}$.
\end{lemma}
For the proof of Lemma \ref{lemma:RestrictedSplitting}, we refers directly to the similar proof of \cite[Lemma A.7, pp. 46]{LeTreustTomala17}. We denote by $Q^n(u,z,w)$ the empirical distribution of the sequence $(u^n,z^n,w^n)$ and we denote by $A_{\delta}$ the set of typical sequences with tolerance $\delta>0$, defined by: 
\begin{align}
A_{\delta} = \bigg\{(u^n,z^n,w^n,x^n,y^n), &\quad\text{ s.t. } \quad|| Q^n(u,z,w) - \PP(u,z )\times \QQ(w|u)||_1\leq \delta,\nonumber\\
 &\quad\text{ and } \quad|| Q^n(x,y) - \PP^{\star}(x) \times \mc{T}(y|x)||_1\leq \delta  \bigg\}.
\end{align}
We define $T_{\alpha}(w^n,y^n,z^n)$ and $B_{\alpha,\gamma,\delta}$ depending on parameters $\alpha>0$ and $\gamma>0$:
\begin{align}
T_{\alpha}(w^n,y^n,z^n) =& \bigg\{i \in \{1,\ldots,n\} , \;\;\text{ s.t. }\;\; D\Big(\PP_{\sigma}(U_i|y^n,z^n)\Big|\Big|\QQ(U_i|w_i,z_i)\Big)\leq  \frac{\alpha^2}{2\ln 2} \bigg\},\label{eq:SetTalpha}\\
B_{\alpha,\gamma,\delta}=& \bigg\{ (w^n,y^n,z^n) , \;\;\text{ s.t. } \;\; \frac{|T_{\alpha}(w^n,y^n,z^n)|}{n} \geq 1 - \gamma\;\; \text{ and }\;\;(w^n,y^n,z^n) \in A_{\delta}\bigg\}.\label{eq:SetTalpha}
\end{align}
The notation $B_{\alpha,\gamma,\delta}^c$ stands for the complementary set of $B_{\alpha,\gamma,\delta}\subset \mc{W}^n \times \mc{Y}^n \times \mc{Z}^n$. The sequences $(w^n,y^n,z^n)$ belong to the set $B_{\alpha,\gamma,\delta}$ if: 1) they are typical and 2) if the corresponding posterior belief $\PP_{\sigma}(U_i|y^n,z^n)$ is close in K-L divergence to the target belief $\QQ(U_i|w_i,z_i)$, for a large fraction of stages $i \in \{1,\ldots,n\}$.

The cornerstone of this achievability proof is Proposition \ref{prop:WynerZivCoding}, which refines the analysis of Wyner-Ziv's source coding, by characterizing its posterior beliefs.
\begin{proposition}[Wyner-Ziv's Posterior Beliefs]\label{prop:WynerZivCoding}
If the probability distribution $\PP(u,z)\times \QQ(w|u)$ satisfies:
\begin{align}
\begin{cases}
&\max_{\PP(x)}I(X;Y) - I(U;W|Z) >0,\\
&\mc{V}^{\star}(z,\QQ(u|z,w)) \text{ is a singleton }\forall (z,w)\in \mc{Z}\times\mc{W},
\end{cases}
\end{align}
then
\begin{align}
&\forall \varepsilon>0, \;\forall \alpha>0, \;\forall \gamma>0,\;\exists \bar{\delta}>0,\;\forall \delta< \bar{\delta}, \;\exists \bar{n}\in \N^{\star},\;\forall n\geq  \bar{n},  \exists \sigma, \text{ s.t. } \PP_{\sigma}(B_{\alpha,\gamma,\delta}^c) \leq \varepsilon.\label{eq:propWynerZivCoding}
\end{align}
\end{proposition}
The proof of proposition \ref{prop:WynerZivCoding} is stated in App. \ref{sec:ProofPropositionCode}.
\begin{proposition}\label{prop:UtilityBound}
For any encoding strategy $\sigma$, we have:
\begin{align}
\bigg| \min_{\tau \in \textsf{BR}_{\textsf{d}}(\sigma)}\Phi_{\textsf{e}}^n(\sigma, \tau) - \widehat{\Phi_{\textsf{e}}}\bigg| \leq (\alpha+ 2 \gamma + \delta)\cdot \bar{\phi_{\textsf{e}}} + (1 - \PP_{\sigma}(B_{\alpha,\gamma,\delta})) \cdot \bar{\phi_{\textsf{e}}},
\end{align}
where $ \bar{\phi_{\textsf{e}}} = \max_{u,z,v} \phi_{\textsf{e}}(u,z,v) $ is the largest absolute value of encoder's utility.
\end{proposition}
For the proof of Proposition \ref{prop:UtilityBound}, we refers directly to the similar proof of \cite[Corollary A.18, pp. 53]{LeTreustTomala17}. 

\begin{corollary}\label{coro:AchievabilityProof}
For any $\varepsilon>0$, there exists $\bar{n}\in \N^{\star}$ such that for all $n\geq \bar{n}$ there exists an encoding strategy $\sigma$ such that:
\begin{align}
\bigg| \min_{\tau \in \textsf{BR}_{\textsf{d}}(\sigma)}\Phi_{\textsf{e}}^n(\sigma, \tau) - \widehat{\Phi_{\textsf{e}}}\bigg| \leq \varepsilon.
\end{align}
\end{corollary}
The proof of Corollary \ref{coro:AchievabilityProof} comes from combining Proposition \ref{prop:WynerZivCoding} with Proposition \ref{prop:UtilityBound} and choosing  parameters $\alpha$, $\gamma$, $\delta$ small and  $n\in \N^{\star}$ large. This concludes the achievability proof of Theorem \ref{theo:MaxMinStackelberg}.


\subsection{Proof of Proposition \ref{prop:WynerZivCoding}}\label{sec:ProofPropositionCode}

We assume that the probability distribution $\PP(u,z)\times \QQ(w|u)$ satisfies the two following conditions:
\begin{align}
\begin{cases}
&\max_{\PP(x)}I(X;Y) - I(U;W|Z) >0,\\
&\mc{V}^{\star}(z,\QQ(u|z,w)) \text{ is a singleton }\forall (z,w)\in \mc{Z}\times\mc{W},
\end{cases}
\end{align}
The strict information constraint ensures there exists a small parameter $\eta>0$ and rates  $\textsf{R}\geq 0 $, $\textsf{R}_{\textsf{L}}\geq 0 $, such that:
\begin{eqnarray}
\textsf{R}  + \textsf{R}_{\textsf{L}}& =&       I( U;W )  + \eta  \label{eq:AchievabilityB1} , \\
\textsf{R}_{\textsf{L}}  &\leq &       I( Z;W )  - \eta  \label{eq:AchievabilityB2} , \\
\textsf{R} \; & \leq&   \max_{\PP(x)} I( X; Y )  -  \eta  \label{eq:AchievabilityB3}  .
\end{eqnarray}
We now recall the random coding construction of Wyner-Ziv \cite{wyner-it-1976} and we investigate the corresponding posterior beliefs. We note by $\Sigma$ the random encoding/decoding, defined as follows:
\begin{itemize}
\item[$\bullet$] \textit{Random codebook.} We defines the indices $m\in\mc{M}$ with $| \mc{M}|= 2^{n   \sf{R}} $ and $l\in\mc{M}_{\textsf{L}}$ with $| \mc{M}_{\textsf{L}}  |= 2^{n   \sf{R}_{\textsf{L}} } $. We draw $| \mc{M} \times \mc{M}_{\textsf{L}}   |= 2^{n  ( \sf{R}  + \textsf{R}_{\textsf{L}})  } $ sequences $W^n(m,l)$ with the i.i.d. probability distribution $\QQ^{\otimes n}(w) $, and $| \mc{M}  |= 2^{n   \sf{R}   } $ sequences $X^n(m)$,  with the i.i.d. probability distribution $\PP^{\otimes n}(x) $ that maximizes the channel capacity in \eqref{eq:AchievabilityB3}. 
\item[$\bullet$] \textit{Encoding function.} The encoder observes the sequence of symbols of source $U^n \in  \mc{U}^n$ and finds a pair of indices $(m,l)\in \mc{M} \times  \mc{M}_{\textsf{L}}$ such that the sequences  $\big(U^n,W^n(m,l)\big) \in A_{\delta}$ are jointly typical. It sends the sequence $X^n(m)$ corresponding to the index $m\in \mc{M}$.
\item[$\bullet$] \textit{Decoding function.} The decoder observes the sequence of channel output $Y^n\in\mc{Y}^n$. It returns an index $\hat{m}\in \mc{M}$ such that the sequences  $\big(Y^n,X^n(\hat{m})\big) \in A_{\delta}$ are jointly typical.  Then it observes the sequence of state information $Z^n\in\mc{Z}^n$ and returns an index $\hat{l}\in \mc{M}_{\textsf{L}}$ such that the sequences  $\big(Z^n,W^n(\hat{m},\hat{l})\big) \in A_{\delta}$ are jointly typical. \item[$\bullet$] \textit{Error Event.} We introduce the event of error $E_{\delta} \in \{0,1\}$ defined as follows:
\end{itemize}
\begin{align}
E_{\delta} = \Bigg\{
\begin{array}{lll}
0&\text{ if }&  (M,L)=( \hat{M},\hat{L})  \;\; \text{ and }\;  \big(U^n ,  Z^n, W^n, X^n,Y^n \big)    \in A_{\delta} ,\\
1 &\text{ otherwise.}& 
\end{array}
\Bigg.
\end{align}

\textit{Expected error probability of the random encoding/decoding $\Sigma$.}
For all $\varepsilon_2>0$, for all $\eta >0$, there exists a $\bar{\delta}>0$, for all $\delta \leq \bar{\delta}$  there exists $\bar{n}$ such that for all $n\geq\bar{n}$, the expected probability of the following error events are bounded by $\varepsilon_2$:
\begin{align}
&\E_{\Sigma}\bigg[ \PP\bigg( \forall  (m,l) ,\quad 
\big(U^n, W^n(m,l) \big) \notin A_{\delta} \bigg)\bigg]  \leq \varepsilon_2, \label{eq:AchievProbaB1} \\
&\E_{\Sigma}\bigg[ \PP\bigg(  \exists l'\neq  l  ,\text{ s.t. } 
\big(Z^n , W^n(m,l') \big) \in A_{\delta}\bigg)\bigg]   \leq \varepsilon_2, \label{eq:AchievProbaB2}\\
&\E_{\Sigma}\bigg[ \PP\bigg(  \exists m'\neq  m  ,\text{ s.t. } 
\big(Y^n , X^n(m') \big) \in A_{\delta}\bigg)\bigg]   \leq \varepsilon_2, \label{eq:AchievProbaB3}
\end{align}
Eq. \eqref{eq:AchievProbaB1} comes from \eqref{eq:AchievabilityB1} and the covering lemma \cite[pp. 208]{ElGammalKim(book)11}.\\
Eq. \eqref{eq:AchievProbaB2} comes from \eqref{eq:AchievabilityB2} and the packing lemma \cite[pp. 46]{ElGammalKim(book)11}.\\
Eq. \eqref{eq:AchievProbaB3} comes from \eqref{eq:AchievabilityB3} and the packing lemma \cite[pp. 46]{ElGammalKim(book)11}.

There exists a coding strategy $\sigma$ with small error probability:
\begin{align}
&\forall \varepsilon_2>0,\;  \forall \eta>0, \;  \exists \bar{\delta}>0,\;\forall \delta\leq \bar{\delta},  \; \exists \bar{n}>0,\;\forall n\geq \bar{n},\quad \exists \sigma,\qquad  \PP_{\sigma}\big(E_{\delta}=1 \big) \leq \varepsilon_2. \label{eq:BoundError0}
\end{align}

\textit{Control of the posterior beliefs.} We assume that the event $E_{\delta}=0$ is realized  and we investigate the posterior beliefs $\PP_{\sigma}(u_i|y^n,z^n,E_{\delta}=0)$ induced by Wyner-Ziv's encoding strategy $\sigma$.
\begin{align}
&  \E_{\sigma} \Bigg[ \frac{1}{n}  \sum_{i=1}^n D\bigg(  \PP_{\sigma}(U_i|Y^n,Z^n,E_{\delta}=0) \bigg| \bigg|   \QQ(U_i|W_i,Z_i) \bigg)\Bigg] \nonumber \\
=& \sum_{(w^n,z^n,y^n)\in A_{\delta} }\PP_{\sigma}(w^n,z^n,y^n|E_{\delta}=0) \times \frac{1}{n}  \sum_{i=1}^n D\bigg(  \PP_{\sigma}(U_i|y^n,z^n,E_{\delta}=0) \bigg| \bigg|   \QQ(U_i|w_i,z_i) \bigg) \label{eq:Beliefs1} \\
=&\frac{1}{n}   \sum_{(u^n,z^n,w^n,y^n)\in A_{\delta} }\PP_{\sigma}(u^n,z^n,w^n,y^n|E_{\delta}=0)    \times \log_2 \frac{1}{\prod_{i=1}^n \QQ(u_i|w_i,z_i)}  -   \frac{1}{n}  \sum_{i=1}^n H(U_i|Y^n,Z^n,E_{\delta}=0) \nonumber \\&& \label{eq:Beliefs2} \\
\leq&H(U|W,Z)   - \frac{1}{n}  H(U^n|W^n,Y^n,Z^n,E_{\delta}=0) + \delta  \label{eq:Beliefs3} \\
\leq&H(U|W,Z)  - \frac{1}{n}  H(U^n|W^n,Z^n,E_{\delta}=0) + \delta \label{eq:Beliefs4} \\
=&H(U|W,Z) -  \frac{1}{n}  H(U^n|E_{\delta}=0) +  \frac{1}{n}  I(U^n;W^n|E_{\delta}=0) \nonumber\\
+& \frac{1}{n}  H(Z^n|W^n,E_{\delta}=0) - \frac{1}{n}  H(Z^n|U^n,W^n,E_{\delta}=0) + \delta. \label{eq:Beliefs5} 
\end{align}
Eq. \eqref{eq:Beliefs1}-\eqref{eq:Beliefs2} come from the hypothesis $E_{\delta}=0$ of typical sequences $(u^n,z^n,w^n,y^n)\in A_{\delta} $ and the definition of the conditional K-L divergence \cite[pp. 24]{cover-book-2006}.\\
Eq. \eqref{eq:Beliefs3} comes from property of typical sequences \cite[pp. 26]{ElGammalKim(book)11} and the conditioning that reduces entropy.\\
Eq. \eqref{eq:Beliefs4} comes from the Markov chain $Z^n -\!\!\!\!\minuso\!\!\!\!- U^n -\!\!\!\!\minuso\!\!\!\!- W^n -\!\!\!\!\minuso\!\!\!\!- Y^n$ induced by the strategy $\sigma$, that implies $H(U^n|W^n,Z^n,E_{\delta}=0) =H(U^n|W^n,Y^n,Z^n,E_{\delta}=0)$.\\
Eq. \eqref{eq:Beliefs5} is a reformulation of \eqref{eq:Beliefs4}.\\
\begin{align}
&\frac{1}{n}  H(U^n|E_{\delta}=0)\geq H(U)  - \frac{1}{n}   - \log_2 |\mc{U}| \cdot \PP_{\sigma}\big(E_{\delta}=1 \big),  \label{eq:Control1}\\
&\frac{1}{n}  I(U^n;W^n|E_{\delta}=0) \leq \textsf{R}  + \textsf{R}_{\textsf{L}} =       I( U;W )  + \eta, \label{eq:Control2}\\
&\frac{1}{n}  H(Z^n|W^n,E_{\delta}=0) \leq \frac{1}{n} \log_2 |A_{\delta}(z^n|w^n)| \leq  H(Z|W) + \delta, \label{eq:Control3}\\
&\frac{1}{n}   H(Z^n|U^n,W^n,E_{\delta}=0) \geq H(Z|U,W) - \frac{1}{n}   - \log_2 |\mc{U}| \cdot \PP_{\sigma}\big(E_{\delta}=1 \big).  \label{eq:Control4}
\end{align}
Eq. \eqref{eq:Control1} comes from the i.i.d. source and Fano's inequality.\\
Eq. \eqref{eq:Control2} comes from the cardinality of codebook given by \eqref{eq:AchievabilityB1}. This argument is also used in \cite[Eq. (23)]{MerhavShamai(StateMasking)07}.\\
Eq. \eqref{eq:Control3} comes from the cardinality of $A_{\delta}(z^n|w^n)$, see also \cite[pp. 27]{ElGammalKim(book)11}.\\
Eq. \eqref{eq:Control4} comes from Fano's inequality $H(Z^n|U^n,W^n)$, and the Markov chain $Z^n -\!\!\!\!\minuso\!\!\!\!- U^n -\!\!\!\!\minuso\!\!\!\!- W^n$ $H(Z^n|U^n)$, the i.i.d. property of the source $(U,Z)$ that implies $H(Z|U)$ and the Markov chain $Z -\!\!\!\!\minuso\!\!\!\!- U -\!\!\!\!\minuso\!\!\!\!- W$ that implies  $H(Z|U,W)$.

Equations \eqref{eq:Beliefs5}-\eqref{eq:Control4} shows that on average, the posterior beliefs $ \PP_{\sigma}(u_i|y^n,z^n,E_{\delta}=0)$ induced by strategy $\sigma$ is close to the target probability distribution $\QQ(u|w,z)$.
\begin{align}
&  \E_{\sigma} \Bigg[ \frac{1}{n}  \sum_{i=1}^n D\bigg(  \PP_{\sigma}(U_i|Y^n,Z^n,E_{\delta}=0) \bigg| \bigg|   \QQ(U_i|W_i,Z_i) \bigg)\Bigg] \nonumber \\
\leq&2\delta + \eta + \frac{2}{n}   + 2 \log_2 |\mc{U}| \cdot \PP_{\sigma}\big(E_{\delta}=1 \big) := \epsilon. \label{eq:ControlFinal} 
\end{align}

Then we have: 
\begin{align}
\PP_\sigma(B^c_{\alpha,\gamma,\delta})=&1 - \PP_\sigma(B_{\alpha,\gamma,\delta})   \nonumber\\
=&\PP_\sigma(E_{\delta}=1) \PP_\sigma(B^c_{\alpha,\gamma,\delta}| E_{\delta}=1)  + \PP_\sigma(E_{\delta}=0) \PP_\sigma(B^c_{\alpha,\gamma,\delta}| E_{\delta}=0) \nonumber\\
\leq&\PP_\sigma(E_{\delta}=1)   +  \PP_\sigma(B^c_{\alpha,\gamma,\delta}| E_{\delta}=0) \nonumber\\
\leq&\varepsilon_2  +  \PP_\sigma(B^c_{\alpha,\gamma,\delta}| E_{\delta}=0) .\label{eq:ErrorTerm2}
\end{align}
Moreover:
\begin{align}
& \PP_\sigma(B^c_{\alpha,\gamma,\delta}| E_{\delta}=0)\nonumber\\
=&\sum_{w^n,y^n,z^n}\PP_{\sigma}\Big( (w^n,y^n,z^n)\in  B^c_{\alpha,\gamma,\delta} \Big| E_{\delta}=0\Big)  \label{eq:MarkovIneqB0} \\
=&\sum_{w^n,y^n,z^n}\PP_{\sigma}\Bigg( (w^n,y^n,z^n)\,\quad \text{ s.t. } \quad  \frac{|T_\alpha(w^n,y^n,z^n)|}{n}< 1-\gamma  \Bigg| E_{\delta}=0\Bigg)  \label{eq:MarkovIneqB1} \\
=& \PP_{\sigma}\Bigg( \frac{1}{n} \cdot \bigg|\bigg\{i , \text{ s.t. } D\Big(\PP_{\sigma}(U_i|y^n,z^n)\Big|\Big|\QQ(U_i|w_i,z_i)\Big)\leq  \frac{\alpha^2}{2\ln 2}   \bigg\}\bigg| < 1 -\gamma \Bigg| E_{\delta}=0 \Bigg) \label{eq:MarkovIneqB2}\\
=& \PP_{\sigma}\Bigg( \frac{1}{n} \cdot \bigg| \bigg\{i , \text{ s.t. } D\Big(\PP_{\sigma}(U_i|y^n,z^n)\Big|\Big|\QQ(U_i|w_i,z_i)\Big)>  \frac{\alpha^2}{2\ln 2}   \bigg\}\bigg| \geq \gamma \Bigg| E_{\delta}=0 \Bigg)  \label{eq:MarkovIneqB2b}\\
\leq& \frac{2\ln 2}{\alpha^2\gamma}   \cdot \E_{\sigma}\bigg[  \frac{1}{n}  \sum_{i=1}^n   D\Big(\PP_{\sigma}(U_i|y^n,z^n)\Big|\Big|\QQ(U_i|w_i,z_i)\Big)\bigg] \label{eq:MarkovIneqB3}  \\
\leq&\frac{2\ln 2}{\alpha^2\gamma}   \cdot \bigg( \eta +  \delta + \frac{2}{n}   +2 \log_2 |\mc{U}| \cdot \PP_{\sigma}\big(E_{\delta}=1 \big)  \bigg) \label{eq:MarkovIneqB4}.
\end{align}
Eq. \eqref{eq:MarkovIneqB0} to \eqref{eq:MarkovIneqB2b} are simple reformulations.\\
Eq. \eqref{eq:MarkovIneqB3} comes from the double use of Markov's inequality as in \cite[Lemma A.22, pp.60]{LeTreustTomala17}. \\
Eq. \eqref{eq:MarkovIneqB4} comes from \eqref{eq:ControlFinal}.

Combining equations \eqref{eq:BoundError0}, \eqref{eq:ErrorTerm2}, \eqref{eq:MarkovIneqB4} and choosing $\eta>0$ small, we obtain the following statement:
\begin{align}
&\forall \varepsilon>0, \;\forall \alpha>0, \;\forall \gamma>0,\;\exists \bar{\delta}>0,\;\forall \delta< \bar{\delta}, \;\exists \bar{n}\in \N^{\star},\;\forall n\geq  \bar{n},  \exists \sigma, \text{ s.t. } \PP_{\sigma}(B_{\alpha,\gamma,\delta}^c) \leq \varepsilon.\label{eq:propWynerZivCoding}
\end{align}
This concludes the proof of Proposition \ref{prop:WynerZivCoding}.




\section{Converse Proof of Theorem \ref{theo:MaxMinStackelberg}}\label{sec:ConverseProof}

We consider an encoding strategy $\sigma$ of length $n\in\N$. We denote by $T$ the uniform random variable $\{1,\ldots,n\}$ and the notation $Z^{-T}$ stands for $(Z_1,\ldots,Z_{t-1},Z_{t+1},\ldots Z_n)$, where $Z_T$ has been removed. We introduce the auxiliary random variable $W=(Y^n,Z^{-T},T)$ whose joint probability distribution $\PP(u,z,w)$ with $(U,Z)$ is defined by:
\begin{align}
\PP(u,z,w) =& \PP_{\sigma}\big(u_T,z_T,y^n,z^{-T},T\big) \nonumber \\
=& \PP(T=i) \cdot  \PP_{\sigma}\big(u_T,z_T,y^n,z^{-T}\big|T=i\big) \nonumber \\
=& \frac{1}{n} \cdot  \PP_{\sigma}\big(u_i, z_i,y^n,z^{-i}\big) . \label{eq:distributionW}
\end{align}  
This identification ensures that the Markov chain $W -\!\!\!\!\minuso\!\!\!\!- U_T -\!\!\!\!\minuso\!\!\!\!- Z_T$ is satisfied. Let us fix a decoding strategy $\tau(v^n|y^n,z^n)$ and define $\tilde{\tau}(v|w,z) = \tilde{\tau}(v|y^n,z^{-i},i,z) = \tau_i(v_i|y^n,z^n)$ where $\tau_i$ denotes the $i$-th coordinate of  $\tau(v^n|y^n,z^n)$. The encoder's long-run utility writes:
\begin{align}
\Phi_{\textsf{e}}^n(\sigma,\tau)
=& \sum_{u^n,z^n,y^n}\PP_{\sigma}(u^n,z^n,y^n )  \sum_{v^n} \tau(v^n| y^n,z^n )  \cdot \Bigg[    \frac{1}{n} \sum_{i=1}^n \phi_{\textsf{e}}(u_i,z_i,v_i)\Bigg] \label{eq:Reformulation2} \\
=& \sum_{i=1}^n  \sum_{u_i,z_i,\atop z^{-i},y^n}  \frac{1}{n}  \cdot \PP_{\sigma}(u_i,z^n,y^n )  \sum_{v_i}\tau_i(v_i| y^n ,z^n)  \cdot     \phi_{\textsf{e}}(u_i,z_i,v_i)\label{eq:Reformulation3} \\
=& \sum_{u_i,z_i,y^n,\atop z^{-i},i} \PP_{\sigma}(u_i,z_i,y^n,z^{-i},i )  \sum_{v_i}\tau_i(v_i| z_i, y^n,z^{-i},i )  \cdot     \phi_{\textsf{e}}(u_i,z_i,v_i) \label{eq:Reformulation4} \\
=& \sum_{u,z,w} \PP(u,z,w )  \sum_{v}\tilde{\tau}(v| w,z ) \cdot     \phi_{\textsf{e}}(u,z,v).\label{eq:Reformulation5}
\end{align}
Eq. \eqref{eq:Reformulation2} - \eqref{eq:Reformulation4} are reformulations and re-orderings.\\
Eq. \eqref{eq:Reformulation5} comes from replacing the random variables $(Y^{n},Z^{-T},T)$ by $W$ whose distribution is defined in \eqref{eq:distributionW}.

Equations \eqref{eq:Reformulation2} - \eqref{eq:Reformulation5} are also valid for the decoder's utility $\Phi_{\textsf{d}}^n(\sigma,\tau)= \sum_{u,z,\atop w,v} \PP(u,z,w )\tilde{\tau}(v| w,z ) \cdot     \phi_{\textsf{d}}(u,z,v)$. A best-reply strategy $\tau\in \textsf{BR}_{\textsf{d}}(\sigma)$ reformulates as:
\begin{align}
& \tau \in \argmax_{\tau'(v^n|y^n,z^n)} \sum_{u^n,z^n,\atop x^n,y^n,v^n} \PP_{\sigma}(u^n,z^n,x^n,y^n) \cdot \tau'(v^n|y^n,z^n)\cdot  \Bigg[ \frac{1}{n} \sum_{i=1}^n \phi_{\textsf{d}}(u_i,z_i,v_i)\Bigg]\\
\Longleftrightarrow&\tilde{\tau}(v|w,z) \in \argmax_{\tilde{\tau}'(v|w,z)}  \sum_{u,z, w} \PP(u,z,w) \cdot \tilde{\tau}'(v|w,z) \cdot   \phi_{\textsf{d}}(u,z,v)\\
\Longleftrightarrow&\tilde{\tau}(v|w,z) \in \Q_2\big(\PP(u,z,w)\big).\label{eq:Identification}
\end{align}  
We now prove that the distribution $\PP(u,z,w)$ defined in \eqref{eq:distributionW}, satisfies the information constraint of the set ${\Q}_0$.
\begin{align}
0 \leq& I(X^n ; Y^n) - I(U^n,Z^n;Y^n) \label{eq:ConverseW1} \\
\leq& \sum_{i=1}^n H( Y_i)  -  \sum_{i=1}^n H(Y_i | X_i) -    I(U^n;Y^n | Z^n)   \label{eq:ConverseW2} \\
\leq& n \cdot \max_{\PP(x)} I(X ; Y)  -    \sum_{i=1}^n   I(U_i;Y^n | Z^n,U^{i-1}) \label{eq:ConverseW3} \\
=& n \cdot \max_{\PP(x)} I(X ; Y)  -    \sum_{i=1}^n   I(U_i;Y^n, Z^{-i},U^{i-1} | Z_i) \label{eq:ConverseW4} \\
\leq& n \cdot \max_{\PP(x)} I(X ; Y)  -    \sum_{i=1}^n   I(U_i;Y^n, Z^{-i} | Z_i) \label{eq:ConverseW5} \\
=& n \cdot \max_{\PP(x)} I(X ; Y)  -   n \cdot     I(U_T;Y^n, Z^{-T} | Z_T,T) \label{eq:ConverseW6} \\
=& n \cdot \max_{\PP(x)} I(X ; Y)  -   n \cdot     I(U_T;Y^n, Z^{-T},T | Z_T) \label{eq:ConverseW7} \\
=& n \cdot \max_{\PP(x)} I(X ; Y)  -   n \cdot     I(U;W | Z) \label{eq:ConverseW8} \\
=& n \cdot  \bigg(\max_{\PP(x)} I(X ; Y)  -  I(U ; W) + I(Z ; W) \bigg) \label{eq:ConverseW9} .
\end{align}
Eq. \eqref{eq:ConverseW1} comes from the Markov chain $Y^n  -\!\!\!\!\minuso\!\!\!\!- X^n   -\!\!\!\!\minuso\!\!\!\!- (U^n,Z^n)$.\\
Eq. \eqref{eq:ConverseW2} comes from the memoryless property of the channel and from removing the positive term $I(U^n; Z^n)\geq0$.\\
Eq. \eqref{eq:ConverseW3} comes from taking the maximum $\PP(x)$ and chain rule.\\
Eq. \eqref{eq:ConverseW4} comes from the i.i.d. property of the source $(U,Z)$ that implies $I(U_i,Z_i;Z^{-i},U^{i-1})=I(U_i;Z^{-i},U^{i-1} | Z_i)=0$.\\
Eq. \eqref{eq:ConverseW5} comes from removing $I(U_i;U^{i-1} | Y^n, Z^{-i},Z_i)\geq0$.\\
Eq. \eqref{eq:ConverseW6} comes from the uniform random variable $T\in\{1,\ldots,n\}$.\\
Eq. \eqref{eq:ConverseW7} comes from the independence between $T$ and the source $(U,Z)$, that implies $I(U_T,Z_T;T)  = I(U_T;T|Z_T)  =0$.\\
Eq. \eqref{eq:ConverseW8} comes from the identification $W = (Y^{n},Z^{-T},T)$. \\
Eq. \eqref{eq:ConverseW9} comes from the  Markov chain $W -\!\!\!\!\minuso\!\!\!\!- U_T -\!\!\!\!\minuso\!\!\!\!- Z_T$. This proves that the distribution $\PP_{\sigma}(u,z,w)$  belongs to the set $\Q_0$.

Therefore, for any encoding strategy $\sigma$ and all $n$, we have:
\begin{align}
&\min_{\tau \in \textsf{BR}_{\textsf{d}}(\sigma)}  \Phi_{\textsf{e}}^n(\sigma,\tau)  \\
=&\min_{\tilde{\tau}(v|w,z) \in \atop \Q_2(\PP(u,z,w))}   \sum_{u,z,w} \PP(u,z,w )  \sum_{v}\tilde{\tau}(v| w,z ) \cdot     \phi_{\textsf{e}}(u,z,v)\\
=&\min_{\tilde{\tau}(v |z,w) \in  \atop \Q_2(\PP(u,z,w))} \E_{\PP(u,z,w) \atop \times \tilde{\tau}(v |z,w)} \bigg[\phi_{\textsf{e}}(U,Z,V)\bigg]\\
\leq&\sup_{ \QQ(u,z,w) \in \Q_0} \min_{\QQ(v |z,w) \in  \atop \Q_2(\QQ(u,z,w))} \E_{\QQ(u,z,w) \atop \times \QQ(v |z,w)} \bigg[\phi_{\textsf{e}}(U,Z,V)\bigg]= \Phi_{\textsf{e}}^{\star}.
\end{align}
The last inequality comes from the probability distribution $\PP(u,z,w)$ that satisfies the information constraint of the set ${\Q}_0$. The first cardinality bound $|\mc{W}| = |\mc{U}|+1$ comes from \cite[Lemma 6.1]{LeTreustTomala17}. The second cardinality bound for $|\mc{W}| = |\mc{V}|^{|\mc{Z}|}$ comes from \cite[Lemma 6.3]{LeTreustTomala17}, by considering the encoder tells to the decoder which action $v\in \mc{V}$ to play in each state $z\in \mc{Z}$.


This conclude the proof of \eqref{eq:Converse} in Theorem \ref{theo:MaxMinStackelberg}.



\section{Proof of Lemma \ref{lemma:FeasiblePosteriors}}\label{sec:ProofLemmaPosteriors}

By inverting the system of equations, we have the following equivalence:
\begin{align}
\begin{cases}
q_1 &= \frac{p_0 \cdot \beta}{p_0 \cdot \beta + (1-p_0) \cdot (1- \alpha)}\\
q_2 &= \frac{p_0 \cdot (1-\beta)}{p_0 \cdot (1-\beta)  + (1-p_0) \cdot \alpha}
\end{cases}
&\Longleftrightarrow
\begin{cases}
\alpha &= \frac{(1 - q_2) \cdot (q_1   - p_0 )}{(1-p_0) \cdot (q_1 - q_2) }\\
\beta &= \frac{q_1 \cdot ( p_0  - q_2 ) }{   p_0 \cdot(q_1 - q_2) }.\\
\end{cases}
\end{align}
Assume that  $q_1   \leq q_2$, then we have:
\begin{align}
0\leq \alpha &\Longleftrightarrow 0\leq \frac{(1 - q_2) \cdot (q_1   - p_0 )}{(1-p_0) \cdot (q_1 - q_2) } \Longleftrightarrow q_1   - p_0 \leq 0,\\
 \alpha \leq 1 &\Longleftrightarrow  \frac{(1 - q_2) \cdot (q_1   - p_0 )}{(1-p_0) \cdot (q_1 - q_2) } \leq 1 \Longleftrightarrow q_2   - p_0 \geq 0,\\
 0\leq \beta &\Longleftrightarrow 0\leq \frac{q_1 \cdot ( p_0  - q_2 ) }{   p_0 \cdot(q_1 - q_2) } \Longleftrightarrow p_0  - q_2 \leq 0,\\
 \beta \leq 1 &\Longleftrightarrow  \frac{q_1 \cdot ( p_0  - q_2 ) }{   p_0 \cdot(q_1 - q_2) } \leq 1 \Longleftrightarrow  p_0  - q_1\geq 0.\\
\end{align}
This proves the equivalence:
\begin{align}
\begin{cases}
q_1   \leq q_2\\
\alpha \in [0,1]\\
\beta \in [0,1]
\end{cases}
&\Longleftrightarrow q_1 \leq p_0 \leq q_2.
\end{align}
Assume that  $q_1   \geq q_2$, then we have:
\begin{align}
0\leq \alpha &\Longleftrightarrow 0\leq \frac{(1 - q_2) \cdot (q_1   - p_0 )}{(1-p_0) \cdot (q_1 - q_2) } \Longleftrightarrow q_1   - p_0 \geq 0,\\
 \alpha \leq 1 &\Longleftrightarrow  \frac{(1 - q_2) \cdot (q_1   - p_0 )}{(1-p_0) \cdot (q_1 - q_2) } \leq 1 \Longleftrightarrow q_2   - p_0 \leq 0,\\
 0\leq \beta &\Longleftrightarrow 0\leq \frac{q_1 \cdot ( p_0  - q_2 ) }{   p_0 \cdot(q_1 - q_2) } \Longleftrightarrow p_0  - q_2 \geq 0,\\
 \beta \leq 1 &\Longleftrightarrow  \frac{q_1 \cdot ( p_0  - q_2 ) }{   p_0 \cdot(q_1 - q_2) } \leq 1 \Longleftrightarrow  p_0  - q_1\leq 0.\\
\end{align}
This proves the equivalence:
\begin{align}
\begin{cases}
q_1   \geq q_2\\
\alpha \in [0,1]\\
\beta \in [0,1]
\end{cases}
&\Longleftrightarrow q_2 \leq p_0 \leq q_1.
\end{align}
Hence there exists probability parameters $\alpha \in [0,1]$ and $\beta \in [0,1]$ if and only if $q_1 \leq p_0 \leq q_2$ or $q_2 \leq p_0 \leq q_1$.

%



\end{document}